\documentclass[12pt,a4]{amsart}
\usepackage{amsmath}
\usepackage{amssymb}
\usepackage{amsthm}
\usepackage{enumitem}
\usepackage{url}
\usepackage{soul}
\usepackage{times}
\usepackage[T1]{fontenc}
\usepackage[export]{adjustbox}

\usepackage{color}
\usepackage{dsfont}
\usepackage{epsfig}
\usepackage[hidelinks]{hyperref}
\usepackage{setspace}
\usepackage{epstopdf}
\usepackage[authoryear,round,sort]{natbib}
\usepackage{comment}
\usepackage{booktabs}
\usepackage{bm}
\usepackage{float}
\usepackage[linesnumbered,ruled,vlined]{algorithm2e}
\SetKwInput{KwInput}{Input}                
\SetKwInput{KwOutput}{Output} 
\usepackage{stmaryrd}
\usepackage{subcaption,graphicx}
\usepackage{float}
\numberwithin{equation}{section}

\usepackage[font=small,labelfont=bf]{caption}
\DeclareCaptionFont{tiny}{\tiny}
\newcommand{\norm}[1]{\left\lVert#1\right\rVert}
\usepackage{geometry}
\theoremstyle{plain}

\newtheorem{theorem}{Theorem}[section]
\newtheorem{lemma}{Lemma}[section]
\newtheorem{corollary}{Corollary}[section]

\theoremstyle{definition}
\newtheorem{definition}{Definition}[section]
\newtheorem{assumption}{Assumption}[section]

\theoremstyle{remark}
\newtheorem{rk}{Remark}[section]
\expandafter\let\expandafter\oldproof\csname\string\proof\endcsname
\let\oldendproof\endproof
\renewenvironment{proof}[1][\proofname]{%
  \oldproof[\noindent\textbf{#1.} ]%
}{\oldendproof}

\DeclareMathOperator\arctanh{arctanh}
\newcommand{\mcH}{\mathcal{H}}

\newcommand{\E}{\mathbb{E}}
\newcommand{\mP}{\mathbb{P}}
\newcommand{\be}{\begin{equation}}
\newcommand{\ee}{\end{equation}}
\newcommand{\by}{\begin{eqnarray*}}
\newcommand{\ey}{\end{eqnarray*}}

\renewcommand{\leq}{\leqslant}
\renewcommand{\geq}{\geqslant}
\usepackage{xcolor}
\definecolor{dark-red}{rgb}{0.4,0.15,0.15}
\definecolor{dark-blue}{rgb}{0.15,0.15,0.4}
\definecolor{medium-blue}{rgb}{0,0,0.5}
\hypersetup{
    colorlinks, linkcolor={red},
    citecolor={blue}, urlcolor={blue}
}
\allowdisplaybreaks

\begin{document}
\title[Improved Metropolis-Hastings algorithms via landscape modification]{Improved Metropolis-Hastings algorithms via landscape modifcation with applications to simulated annealing and the Curie-Weiss model}
\author{Michael C.H. Choi}
\address{Department of Statistics and Data Science and Yale-NUS College, National University of Singapore, Singapore}
\email{michael.choi@yale-nus.edu.sg}

\date{\today}
\maketitle

\begin{abstract}
In this paper, we propose new Metropolis-Hastings and simulated annealing algorithms on finite state space via modifying the energy landscape. The core idea of landscape modification rests on introducing a parameter $c$, in which the landscape is modified once the algorithm is above this threshold parameter to encourage exploration, while the original landscape is utilized when the algorithm is below the threshold for exploitation purpose. We illustrate the power and benefits of landscape modification by investigating its effect on the classical Curie-Weiss model with Glauber dynamics and external magnetic field in the subcritical regime. This leads to a landscape-modified mean-field equation, and with appropriate choice of $c$ the free energy landscape can be transformed from a double-well into a single-well, while the location of the global minimum is preserved on the modified landscape. Consequently, running algorithms on the modified landscape can improve the convergence to the ground-state in the Curie-Weiss model. In the setting of simulated annealing, we demonstrate that landscape modification can yield improved or even subexponential mean tunneling time between global minima in the low-temperature regime by appropriate choice of $c$, and give convergence guarantee using an improved logarithmic cooling schedule with reduced critical height. We also discuss connections between landscape modification and other acceleration techniques such as Catoni's energy transformation algorithm, preconditioning, importance sampling and quantum annealing. The technique developed in this paper is not only limited to simulated annealing and is broadly applicable to any difference-based discrete optimization algorithm by a change of landscape.
		
	\smallskip
	
	
\end{abstract}



\section{Introduction}

Given a target Gibbs distribution with Hamiltonian function $\mcH$ and temperature $\epsilon$, the Metropolis-Hastings (MH) algorithm is a popular and important Markov chain Monte Carlo algorithm that has been applied to various sampling and optimization problems in a wide range of disciplines including but not limited to Bayesian computation, statistical physics and theoretical computer science. While there are many improved variants of the MH algorithm that have been investigated in the literature, this paper centers around a method that is known as landscape modification. This technique is particularly suitable in the context of stochastic minimization with respect to $\mcH$, either by running the MH algorithm at a low enough temperature $\epsilon$ or by driving the temperature $\epsilon$ to zero as in simulated annealing. The core idea of landscape modification relies on targeting a modified Hamiltonian function instead of $\mcH$, where the modification or the transformation is based upon two parameters, namely $f$ and the threshold parameter $c$. On the part of landscape of $\mcH$ that is below $c$, the original landscape is utilized that allows for exploitation and concentration of the MH algorithm in the low temperature. On the other hand, on the region of landscape of $\mcH$ which is above $c$, the function $f$ is applied to transform this part of landscape to encourage and facilitate exploration of the chain. Precisely, the acceptance-rejection probability on this part is increased due to the transformation, which consequently leads to a higher transition rate of the modified MH chain that promotes exploration.

The gain of landscape modification stems from the perspective of critical height, a notion that measures the difficulty of the landscape in a broad sense. With appropriate tuning of both $f$ and $c$, the critical height is reduced, which consequently leads to improved simulated annealing algorithm or for instance reduced mean crossover time or relaxation time in the Curie-Weiss model that we shall discuss in detail.

We now summarize and highlight the key achievements and contributions of this paper:

\begin{enumerate}
	\item \textbf{Propose and analyze Metropolis-Hastings algorithms with landscape modification}. In Section \ref{sec:MHland}, we first define a new MH algorithm using a modified Hamiltonian function. The corresponding acceptance-rejection probability in the modified MH in general has an integral form, but we show that this integral can be readily calculated upon specializing into various choices of $f$ such as linear, quadratic or square root functions. This is then followed by an example on landscape modification in the Ehrenfest urn with a linear Hamiltonian in Section \ref{subsec:ehrenfest}, where we prove an upper bound on the spectral gap with polynomial dependence on the dimension, whereas the same technique yields an exponential dependence on the dimension for the classical MH. We provide a discussion on possible tuning strategies of $f$ and $c$ in Section \ref{subsec:tuning}. In the final subsection, that is, Section \ref{subsec:connection}, we elaborate on the similarities and differences between landscape modification and other acceleration techniques in Markov chain Monte Carlo literature such as Catoni's energy transformation algorithm \cite{Catoni96, Catoni98}, preconditioning of the Hamiltonian, importance sampling and quantum annealing \cite{WWZ16}.
	
	\item \textbf{Improved mean crossover time and relaxation time in the Curie-Weiss (CW) model with landscape modification}. In Section \ref{sec:CW}, we investigate the effect of landscape modification on the CW model. We first consider the CW model under a fixed external magnetic field in the subcritical regime, where the free energy landscape, as a function of the magnetization, has two local minima. Using a Glauber dynamics with landscape modification, we introduce a new mean-field equation, and with appropriate choice of $c$, the free energy landscape is transformed from a double-well into a single-well while preserving the location of the global minimum on the modified landscape. As a result running algorithms on the modified landscape can accelerate the convergence towards the ground-state. We prove a subexponential mean tunneling time in such setting and discuss related metastability results. Similar results are then extended to the random field CW model.
	
	\item \textbf{Improved simulated annealing algorithm with improved logarithmic cooling schedule}. In Section \ref{sec:discreteSA}, we consider the simulated annealing setting by driving the temperature down to zero. We define a clipped critical height $c^*$ (that depends on the threshold parameter $c$) associated with the improved simulated annealing algorithm, and prove tight spectral gap asymptotics based on this parameter. Consequently, this leads to similar asymptotic results concerning the total variation mixing time and the mean tunneling time in the low-temperature regime. Utilizing existing results concerning simulated annealing with time-dependent target function, we prove convergence guarantee of the landscape modified simulated annealing on finite state space with an improved logarithmic cooling schedule. These theoretical results are corroborated with numerical experiments in Section \ref{subsec:numerical} on a travelling salesman problem that offer numerical evidence to support the improved convergence results.	
\end{enumerate}

This paper can be considered as a sequel to an earlier work by the author \cite{C20KSA}. The original motivation of landscape modification is from \cite{FQG97}, who propose a variant of overdamped Langevin diffusion with state-dependent diffusion coefficient. In \cite{C20KSA}, we cast this idea of state-dependent noise via landscape modification in the setting of kinetic simulated annealing and develop an improved kinetic simulated annealing algorithm with convergence guarantee. In this paper, we recognize that this idea of landscape modification can also be applied to the finite state space setting in Metropolis-Hastings (MH) and simulated annealing, and investigate the benefits and speedups that this technique can bring in particular to the analysis of Curie-Weiss model and stochastic optimization. While the technique developed in \cite{C20KSA} can be readily applied to gradient-based continuous optimization algorithms, we emphasize that the landscape modification technique proposed in this paper can be analogously implemented in essentially all discrete optimization algorithms by a change of landscape and is not limited to simulated annealing or the CW model.

\subsection{Notations}

Throughout this paper, we adopt the following notations. For $x,y \in \mathbb{R}$, we write $x_+ = \max\{x,0\}$ to denote the non-negative part of $x$, and $x \wedge y = \min\{x,y\}$. For two functions $g_1, g_2 : \mathbb{R} \to \mathbb{R}$, we say that $g_1 = \mathcal{O}(g_2)$ if there exists a constant $C > 0$ such that for sufficiently large $x$, we have $|g_1(x)| \leq C g_2(x)$. We write $g_1 = o(g_2)$ if $\lim_{x \to \infty} g_1(x)/g_2(x) = 0$, and denote $g_1 \sim g_2$ if $\lim_{x \to \infty} g_1(x)/g_2(x) = 1$. We say that $g_1(x)$ is a subexponential function if $\lim_{x \to \infty} \frac{1}{x} \log g_1(x) = 0$.

\section{Metropolis-Hastings with landscape modification}\label{sec:MHland}

Let $\mathcal{X}$ be a finite state space under consideration, $Q = (Q(x,y))_{x,y \in \mathcal{X}}$ be the transition matrix of a reversible proposal chain with respect to the probability measure $\mu = (\mu(x))_{x \in \mathcal{X}}$, and $\mathcal{H}:\mathcal{X} \to \mathbb{R}$ be the target Hamiltonian function. Denote by $M^{0} = (M^{0}(x,y))_{x,y \in \mathcal{X}}$ to be the infinitesimal generator of the continuized classical Metropolis-Hastings chain $X^0 = (X^0(t))_{t \geq 0}$, with proposal chain $Q$ and target distribution being the Gibbs distribution $\pi^{0}(x) \propto e^{-\frac{1}{\epsilon} \mathcal{H}(x)}\mu(x)$ at temperature $\epsilon > 0$. Recall that its dynamics is given by
$$M^{0}(x,y) = M^0_{\epsilon}(Q,\pi^0)(x,y) := \begin{cases} Q(x,y) \min \left\{ 1,e^{\frac{1}{\epsilon}(\mathcal{H}(x)-\mathcal{H}(y))} \right\} = Q(x,y) e^{-\frac{1}{\epsilon}(\mcH(y)-\mcH(x))_+}, &\mbox{if } x \neq y; \\
	- \sum_{z: z \neq x} M^0(x,z), & \mbox{if } x = y. \end{cases}$$
We shall explain the upper script of $0$ in both $M^0$ and $\pi^0$ in Definition \ref{def:MHland} below.

Let us denote the ground-state energy level or the global minimum value of $\mcH$ to be $\mcH_{\textrm{min}} := \min_{x \in \mathcal{X}} \mcH(x)$. Instead of targeting directly the Hamiltonian $\mathcal{H}$ in the Gibbs distribution, we instead target the following modified or transformed function $\mcH_{\epsilon,c}^f$ at temperature $\epsilon$:
\begin{align}\label{eq:mcHeps}
	\mathcal{H}_{\epsilon}(x) = \mathcal{H}^f_{\epsilon,c}(x) := \int_{\mcH_{\textrm{min}}}^{\mathcal{H}(x)} \dfrac{1}{f((u-c)_+) + \epsilon}\,du,
\end{align}
where the function $f$ and the parameter $c$ are chosen to satisfy the following assumptions:
\begin{assumption}\label{assump:main}
	\begin{enumerate}		
		\item\label{it:assumpf} The function $f: \mathbb{R}^+ \to \mathbb{R}^+$ is differentiable, non-decreasing and satisfies
		$$f(0) = 0.$$
		\item $c$ satisfies $c \geq \mcH_{\textrm{min}}$. 
	\end{enumerate}
\end{assumption}
While it is impossible to calculate $\mcH_{\epsilon,c}^f$ without knowing $\mcH_{\textrm{min}}$ a priori, in a Metropolis-Hastings chain what matters is the difference of the energy function. For $x,y \in \mathcal{X}$, we see that
$$\mathcal{H}_{\epsilon,c}^f(y) - \mcH_{\epsilon,c}^f(x) = \int_{\mcH(x)}^{\mathcal{H}(y)} \dfrac{1}{f((u-c)_+) + \epsilon}\,du,$$
which does not depend on $\mcH_{\textrm{min}}$. 
In the special case when we choose $f = 0$, the above equation reduces to $\mcH_{\epsilon,c}^0(y) - \mcH_{\epsilon,c}^0(x) = \frac{1}{\epsilon}(\mcH(y) - \mcH(x))$. On the other hand, in the case where  $c < \mcH(x) < \mcH(y)$ and $f$ is chosen such that $f(z) > 0$ whenever $z > 0$, then we have  
$$\mathcal{H}_{\epsilon,c}^f(y) - \mcH_{\epsilon,c}^f(x) \leq \dfrac{1}{f(\mcH(x)-c) + \epsilon}(\mcH(y) - \mcH(x)) < \frac{1}{\epsilon} (\mcH(y) - \mcH(x)).$$
Since $f$ is assumed to be non-decreasing in Assumption \ref{assump:main}, the greater the difference between $\mcH(x)$ and $c$, the smaller the upper bound in the first inequality in the above equation, the smaller it is we would expect for $\mathcal{H}_{\epsilon,c}^f(y) - \mcH_{\epsilon,c}^f(x)$, the higher the transition rate $Q(x,y)\exp\{-(\mathcal{H}_{\epsilon,c}^f(y) - \mcH_{\epsilon,c}^f(x))\}$, and the landscape is modified in this sense. Intuitively speaking, when the algorithm is above the threshold parameter $c$ the landscape is modified such that the transition rate is higher to encourage exploration, while the original landscape is utilized for exploitation when the algorithm is below $c$. 

To illustrate the effect of the proposed transformation, we plot the following function
$$\mcH(x) =\cos (2 x)+\frac{1}{2} \sin (x)+\frac{1}{3} \sin (10 x)$$
as in \cite{M18} and compare this landscape with that of $\mcH_{\epsilon,c}^f$ in Figure \ref{fig:landscape}, where we take $\epsilon \in \{0.25,0.5,0.75,1\}$, $c = -1.5$ and $f(z) = z$. With these choices of parameters, we compute that
$$\mcH^f_{\epsilon,c}(x) = \frac{1}{\epsilon} (\min\{c,\mcH(x)\} - \min \mcH) + \ln\left(1 + \frac{1}{\epsilon}(\mcH(x) - c)_+\right).$$
In view of the above equation, we plot and compare $\frac{1}{\epsilon} \mcH(x)$ and $\mcH^f_{\epsilon,c}(x) + \frac{1}{\epsilon} \mcH_{\textrm{min}}$ in Figure \ref{fig:landscape} on the domain $\mathcal{X} = \{-5 + \frac{k}{1000} 10;~ k = 1,2,\ldots,1000\}$. The shift of $+\frac{1}{\epsilon}\mcH_{\textrm{min}}$ is necessary to make these two landscapes match exactly in the region where $\{x;~\mcH(x) \leq c\}$ so that they are on the same scale. When $x \in (-4,-2)$ in Figure \ref{fig:landscape}, we can see that the gradient of $\mcH_{\epsilon,c}^f$ is smaller than that of $\frac{1}{\epsilon} \mcH$, thus it is easier to climb up the hill in this region (i.e. higher transition rate). From the plot we also note that both $\frac{1}{\epsilon} \mcH$ (the red and solid curve) and $\mcH_{\epsilon,c}^f$ (the blue and dashed curve) share exactly the same set of stationary points. 
\begin{figure}
	\centering
	{\renewcommand{\arraystretch}{0}
		\begin{tabular}{c@{}c}
			\begin{subfigure}[b]{.475\columnwidth}
				\centering
				\includegraphics[width=\columnwidth]{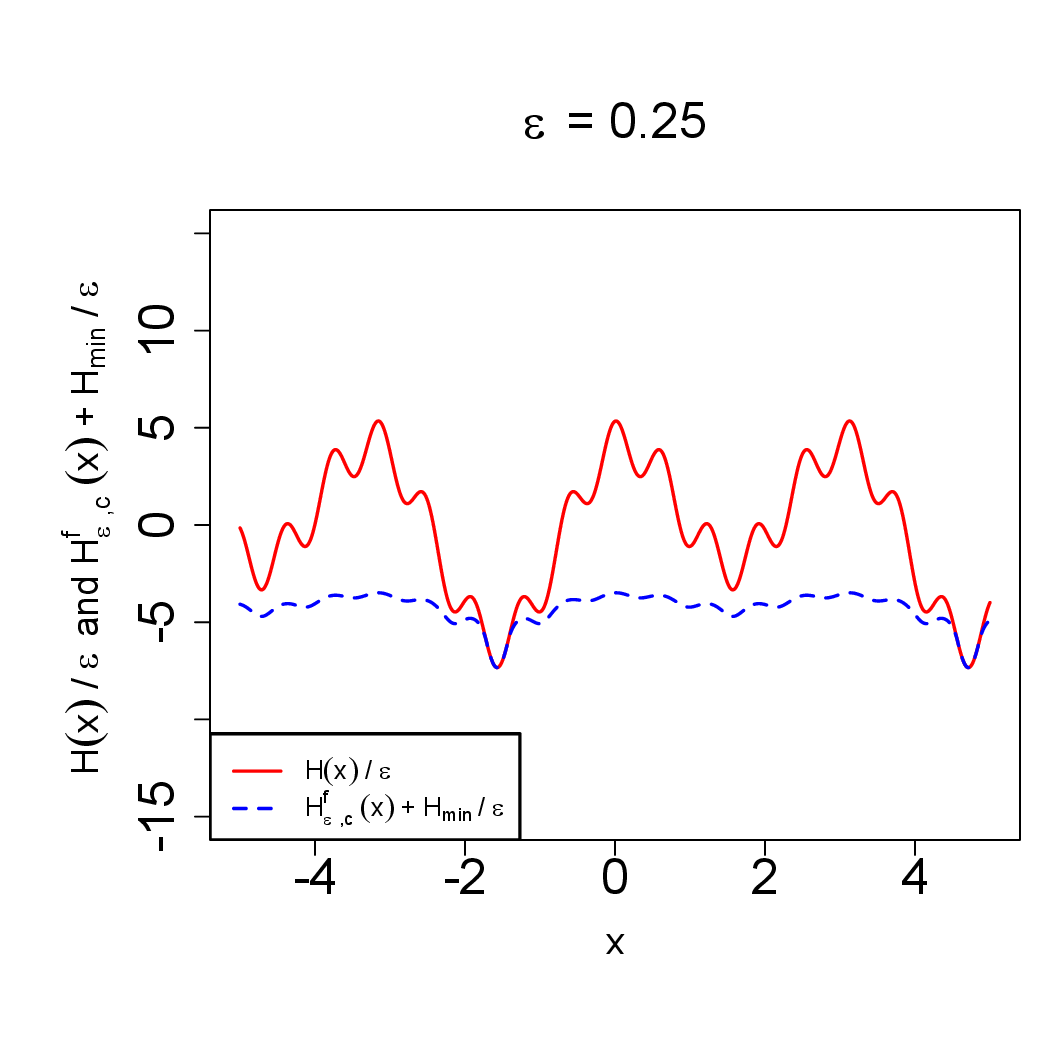}%
				\label{fig:mean and std of net14}
			\end{subfigure}&
			\begin{subfigure}[b]{.475\columnwidth}  
				\centering
				\includegraphics[width=\columnwidth]{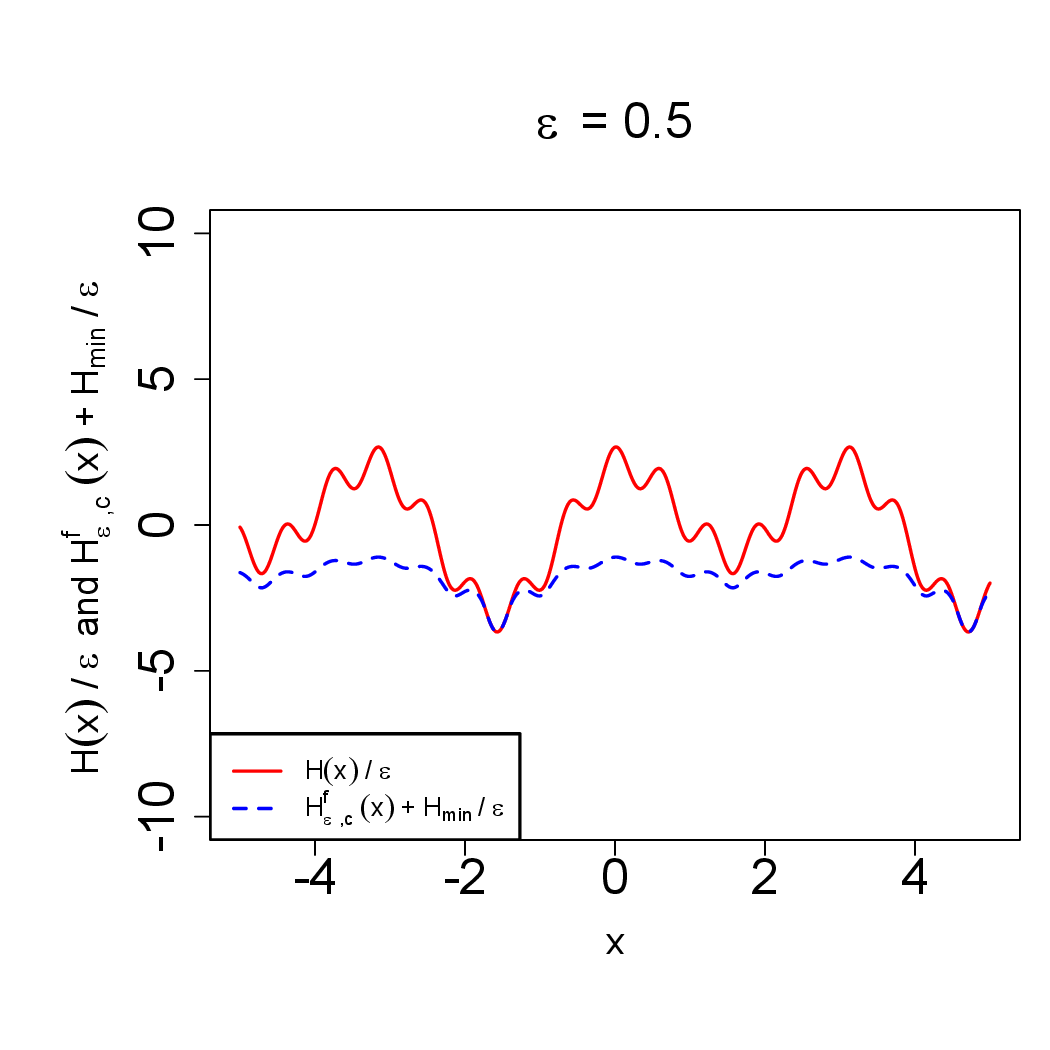}%
				\label{fig:mean and std of net24}
			\end{subfigure}\\
			\begin{subfigure}[t]{.475\columnwidth}   
				\centering 
				\includegraphics[width=\textwidth]{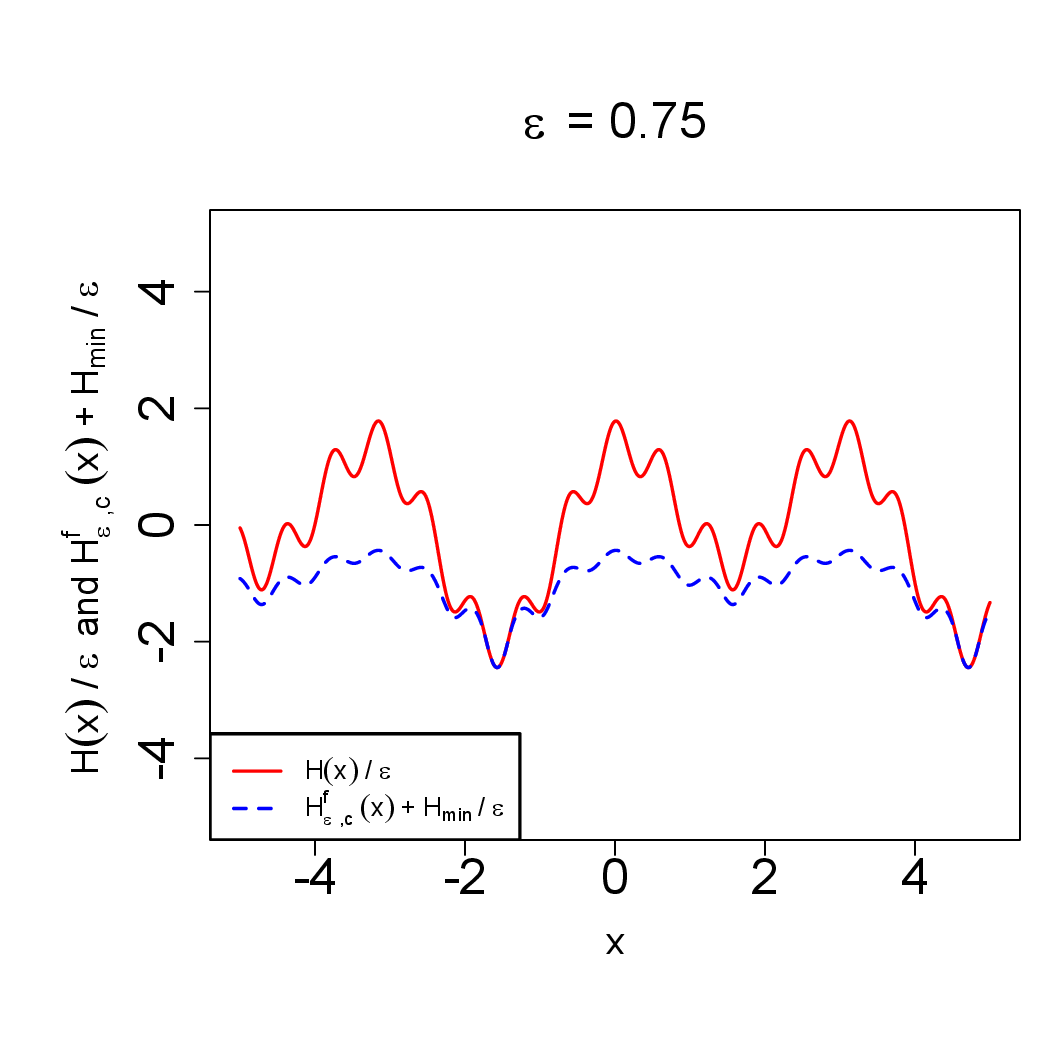}%
				\label{fig:mean and std of net34}
			\end{subfigure}&
			\begin{subfigure}[t]{.475\columnwidth}   
				\centering 
				\includegraphics[width=\columnwidth]{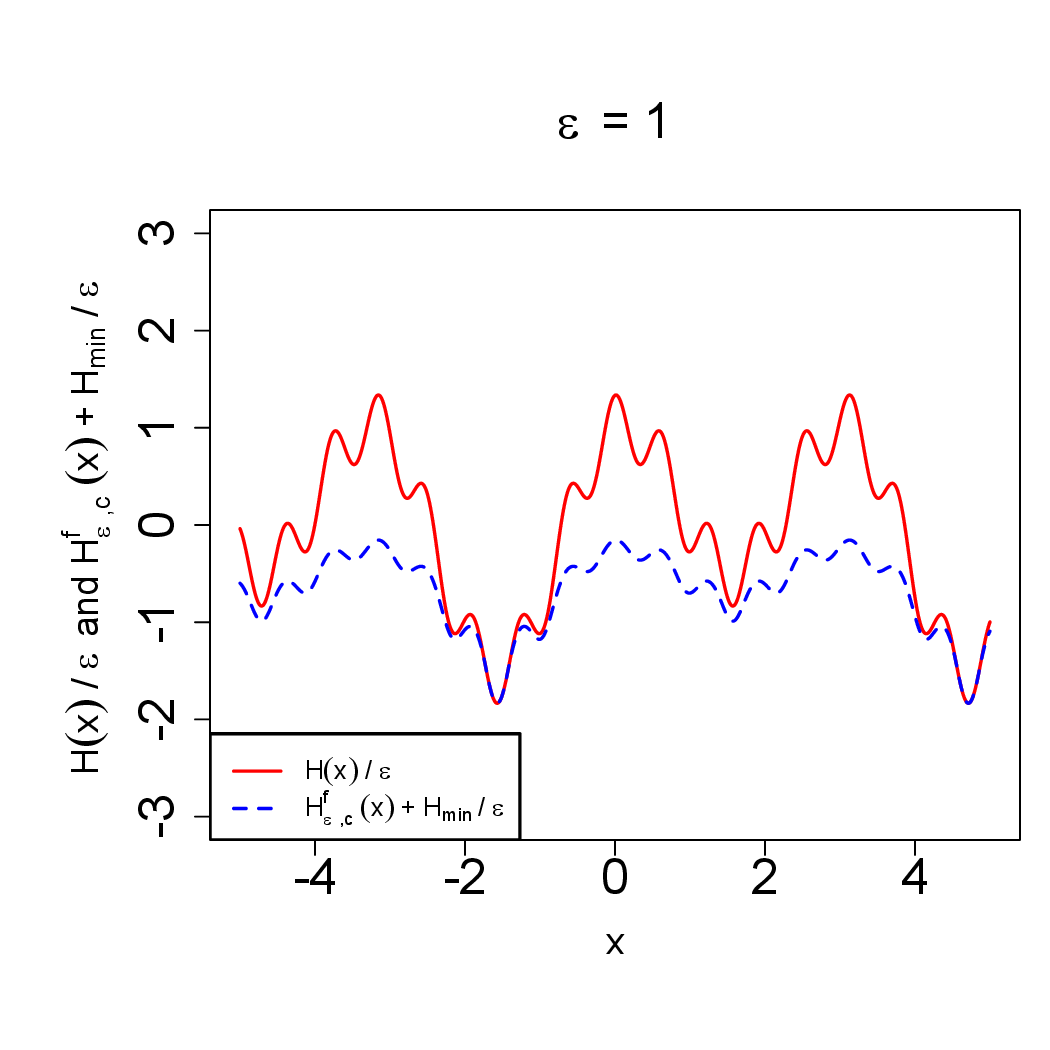}%
				\label{fig:mean and std of net44}
			\end{subfigure}
	\end{tabular}}
	\caption[]
	{\small Landscape of $\frac{1}{\epsilon} \mcH$ and $\mcH^f_{\epsilon,c} + \frac{1}{\epsilon} \mcH_{\textrm{min}}$, where $\mcH(x) =\cos (2 x)+\frac{1}{2} \sin (x)+\frac{1}{3} \sin (10 x)$, $\epsilon \in \{0.25,0.5,0.75,1\}$, $c = -1.5$ and $f(z) = z$.} 
	\label{fig:landscape}
\end{figure}

With these ideas and notations in mind, we are now ready to introduce the Metropolis-Hastings chain with landscape modification:

\begin{definition}[Metropolis-Hastings with landscape modification]\label{def:MHland}
	Let $\mcH$ be the target Hamiltonian function, and $\mcH_{\epsilon,c}^f$ be the landscape-modified function at temperature $\epsilon$ introduced in \eqref{eq:mcHeps}, where $f$ and $c$ satisfy Assumption \ref{assump:main}. The continuized Metropolis-Hastings chain with landscape modification $X^f_{\epsilon,c} = (X^f_{\epsilon,c}(t))_{t \geq 0}$ has target distribution $\pi^f(x) = \pi^f_{\epsilon,c}(x) \propto e^{-\mcH_{\epsilon,c}^f(x)}\mu(x)$, proposal chain $Q$, and its infinitesimal generator $M^f = (M^f(x,y))_{x,y \in \mathcal{X}}$ is given by
	$$M^{f}(x,y) = M_{\epsilon,c}^f(Q,\pi^f)(x,y) := \begin{cases} Q(x,y) e^{-(\mcH_{\epsilon,c}^f(y)-\mcH_{\epsilon,c}^f(x))_+}\, &\mbox{if } x \neq y; \\
		- \sum_{z: z \neq x} M^f(x,z), & \mbox{if } x = y. \end{cases}$$
	Note that when $f = 0$, the above dynamics reduces to the classical Metropolis-Hastings $X^0$.
\end{definition}

We now fix a few notations and recall some important concepts and results in the Markov chains literature. We endow the Hilbert space $\ell^2(\pi^f_{\epsilon,c})$ with the usual inner product weighted by the invariant measure $\pi^f_{\epsilon,c}$: for $g_1,g_2 \in \ell^2 (\pi^f_{\epsilon,c})$,
$$\langle g_1,g_2 \rangle_{\pi^f_{\epsilon,c}} := \sum_{x \in \mathcal{X}} g_1(x) g_2(x) \pi^f_{\epsilon,c}(x),$$
and for $p > 1$ we denote the $\ell^p$ norm by $\norm{\cdot}_{\ell^p(\pi^f_{\epsilon,c})}$.
We write $\lambda_2(-M^f_{\epsilon,c})$ to be the spectral gap of $M^f_{\epsilon,c}$, that is,
\begin{align}\label{eq:spectralgap}
	\lambda_2(-M^f_{\epsilon,c}) := \inf_{l \in \ell^2(\pi^f_{\epsilon,c}): \pi^f_{\epsilon,c}(l) = 0} \dfrac{\langle -M^f_{\epsilon,c}l,l \rangle_{\pi^f_{\epsilon,c}}}{\langle l,l \rangle_{\pi^f_{\epsilon,c}}}.
\end{align} 
Analogously, we write $\lambda_2(-M^0_{\epsilon})$ to denote the spectral gap of $M^0_{\epsilon}$.

In the upcoming sections, we shall investigate and compare the total variation mixing time between the Metropolis-Hastings chains with and without landscape modification. For any probability measure $\nu_1,\nu_2$ with support on $\mathcal{X}$, the total variation distance between $\nu_1$ and $\nu_2$ is 
$$||\nu_1 - \nu_2||_{TV} := \sup_{A \subset \mathcal{X}} |\nu_1(A) - \nu_2(A)| = \dfrac{1}{2} \sum_{x \in \mathcal{X}} |\nu_1(x) - \nu_2(x)|.$$
The worst-case total variation mixing time is defined to be
$$t_{mix}(M^f_{\epsilon,c},1/4) := \inf\left\{t;~\sup_x ||P_t^{f}(x,\cdot) - \pi^f_{\epsilon,c}||_{TV} < 1/4\right\},$$
where $(P_t^{f} = e^{M^f_{\epsilon,c}t})_{t \geq 0}$ is the transition semigroup of $X^f_{\epsilon,c}$. 

Besides the mixing time, we will also be interested in various hitting times of the Metropolis-Hastings chain. These variables naturally appear when we discuss metastability results in the Curie-Weiss model in Section \ref{sec:CW} or in discrete simulated annealing in Section \ref{sec:discreteSA}. For any $A \subset \mathcal{X}$, we denote $\tau^f_{A} := \inf\{t \geq 0; X^f_{\epsilon,c}(t) \in A\}$ and the usual convention that $\inf \emptyset = \infty$ applies. Similarly we define $\tau^0_A$ to be the first hitting time of the set $A$ for the Metropolis-Hastings chain $X^0$. When $A = \{x\}$, we shall simply write $\tau^f_x = \tau^f_{\{x\}}$ (resp.~ $\tau^0_x = \tau^0_{\{x\}}$). Also, we write $\mathbb{E}_x(\cdot)$ (resp.~ $\mathbb{P}_x(\cdot)$) to denote the mathematical expectation (resp.~ probability) of the Markov chain with initial state at $x \in \mathcal{X}$. 

One crucial notion that quantifies possible benefits of landscape modification is the concept of critical height, which in a broad sense measures the difficulty of the landscape. To this end, we recall this classical notion that orginates from the simulated annealing and metastability literature. A path from $x$ to $y$ is any sequence of points starting from $x_0 = x, x_1, x_2,\ldots, x_n =y$ such that $Q(x_{i-1},x_i) > 0$ for $i = 1,2,\ldots,n$. For any $x \neq y$, such path exists as the proposal chain $Q$ is irreducible. We write $\chi^{x,y}$ to be the set of paths from $x$ to $y$, and elements of $\chi^{x,y}$ are denoted by $\gamma = (\gamma_i)_{i=0}^n$. Given a target function $\mathcal{U}$ defined on $\mathcal{X}$, the highest value of $\mathcal{U}$ along a path $\gamma \in \chi^{x,y}$, known as the elevation, is defined to be 
$$\mathrm{Elev}(\mathcal{U},\gamma) := \max\{\mathcal{U}(\gamma_i);~\gamma_i \in \gamma\},$$ 
and the lowest possible highest elevation along path(s) from $x$ to $y$ is 
\begin{align}\label{eq:G}
	G(x,y) = G(\mathcal{U},x,y) := \min\{\mathrm{Elev}(\mathcal{U},\gamma);~\gamma \in \chi^{x,y}\}.
\end{align}
The associated critical height of $\mathcal{U}$ is then defined to be
\begin{align*}
	L(\mathcal{U}) &:= \max_{x,y \in \mathcal{X}}\{G(x,y) - \mathcal{U}(x) - \mathcal{U}(y) \} + \min \mathcal{U}.
\end{align*} 
Another related notion is the clipped critical height $c^*$, which is defined to be
\begin{align*}
	c^*(\mathcal{U},c) := \max_{x,y \in \mathcal{X}} \{(G(\mathcal{U},x,y) \wedge c) - (\mathcal{U}(x) \wedge c) - (\mathcal{U}(y) \wedge c) \} + \min \mathcal{U}.
\end{align*}
One can understand $c^*$ as if we are optimizing with respect to the function $\mathcal{U} \wedge c$. In Section \ref{sec:CW}, we shall consider $\mathcal{U}$ to be the free energy with and without landscape modification that arises in the Curie-Weiss model, while we take $\mathcal{U}$ to be either $\mathcal{H}$ or $\mathcal{H}^f_{\epsilon,c}$ in Section \ref{sec:discreteSA} when we investigate an improved simulated annealing algorithm. We refer readers to Figure \ref{fig:c*} in Section \ref{sec:discreteSA} where we offer a visual illustration of the notion of critical height.

To simulate $X^{f}_{\epsilon,c}$ practically, we would need to evaluate the acceptance-rejection probability, which amounts to the following integration:
\begin{align}\label{eq:arint}
	\exp\left(-(\mcH_{\epsilon,c}^f(y)-\mcH_{\epsilon,c}^f(x))_+\right) &= \begin{cases} 1, &\mbox{if } \mcH(y) \leq \mcH(x); \\
		\exp\left(-\frac{1}{\epsilon}(\mcH(y) - \mcH(x))\right), & \mbox{if } c\geq \mcH(y) > \mcH(x); \\
		\exp\left(-\frac{1}{\epsilon}(c - \mcH(x)) - \int_{c}^{\mcH(y)} \frac{1}{f(u-c) + \epsilon}\,du\right), & \mbox{if } \mcH(y) > c \geq \mcH(x); \\
		\exp\left(-\int_{\mcH(x)}^{\mcH(y)} \frac{1}{f(u-c) + \epsilon}\,du\right), & \mbox{if } \mcH(y) > \mcH(x) > c. 
	\end{cases}
\end{align}

In the following three subsections, we evaluate the above integrals \eqref{eq:arint} where we choose $f$ to be a linear, quadratic, or cubic function, respectively.

\subsection{Linear $f$: Metropolis-Hastings with logarithmic Hamiltonian and Catoni's energy transformation algorithm}\label{subsec:linearf}

In this subsection, we specialize into $f(u) = u$. It turns out we can understand the landscape modification as if the Hamiltonian is on a logarithmic scale whenever $\mcH(x) > c$.

For $x,y \in \{\mcH(y) > \mcH(x) \geq c\}$, since 
$$\int_{\mcH(x)}^{\mcH(y)} \frac{1}{u-c + \epsilon}\,du = \ln \left(\dfrac{\mcH(y)-c+\epsilon}{\mcH(x)-c+\epsilon}\right),$$
putting the expression back into \eqref{eq:arint} gives
\begin{align*}
	\exp\left(-(\mcH_{\epsilon,c}^f(y)-\mcH_{\epsilon,c}^f(x))_+\right) &= \begin{cases} 1, &\mbox{if } \mcH(y) \leq \mcH(x); \\
		\exp\left(-\frac{1}{\epsilon}(\mcH(y) - \mcH(x))\right), & \mbox{if } c\geq \mcH(y) > \mcH(x); \\
		\exp\left(-\frac{1}{\epsilon}(c - \mcH(x))\right) \dfrac{\epsilon}{\mcH(y)-c+\epsilon}, & \mbox{if } \mcH(y) > c \geq \mcH(x); \\
		\dfrac{\mcH(x)-c+\epsilon}{\mcH(y)-c+\epsilon}, & \mbox{if } \mcH(y) > \mcH(x) > c. 
	\end{cases}
\end{align*}

This resulting dynamics $M^f$ coincides with the energy transformation method introduced by \cite{Catoni98, Catoni96} on $\{\mcH(y) > \mcH(x) > c\}$, which is based on logarithmic Hamiltonian. We refer readers to Section \ref{subsec:connection} for a more detailed account on the connection between landscape modification and Catoni's energy transformation.

\subsection{Quadratic $f$: Metropolis-Hastings with $\arctan$ Hamiltonian}\label{subsec:quadf}

In this subsection, we take $f(u) = u^2$. In this case, the effect of landscape modification gives an inverse-tangent-transformed Hamiltonian whenever $\mcH(x) > c$.

For $x,y \in \{\mcH(y) > \mcH(x) \geq c\}$, using the inverse-tangent difference formula we obtain 
\begin{align*}
	\int_{\mcH(x)}^{\mcH(y)} \frac{1}{(u-c)^2 + \epsilon}\,du &= \sqrt{\frac{1}{\epsilon}}\left(\arctan\left(\sqrt{\frac{1}{\epsilon}}(\mcH(y)-c)\right) - \arctan\left(\sqrt{\frac{1}{\epsilon}}(\mcH(x)-c)\right)\right), \\
	&= \sqrt{\frac{1}{\epsilon}} \arctan\left(\dfrac{\sqrt{\frac{1}{\epsilon}}(\mcH(y)-\mcH(x))}{1+\frac{1}{\epsilon}(\mcH(y)-c)(\mcH(x)-c)}\right),
\end{align*}
and substituting the above expression back into \eqref{eq:arint} gives
{\scriptsize
	\begin{align*}
		\exp\left(-(\mcH_{\epsilon,c}^f(y)-\mcH_{\epsilon,c}^f(x))_+\right) &= \begin{cases} 1, &\mbox{if } \mcH(y) \leq \mcH(x); \\
			\exp\left(-\frac{1}{\epsilon}(\mcH(y) - \mcH(x))\right), & \mbox{if } c\geq \mcH(y) > \mcH(x); \\
			\exp\left(-\frac{1}{\epsilon}(c - \mcH(x)) - \sqrt{\frac{1}{\epsilon}}\arctan\left(\sqrt{\frac{1}{\epsilon}}(\mcH(y)-c)\right)\right), & \mbox{if } \mcH(y) > c \geq \mcH(x); \\
			\exp\left(\sqrt{\frac{1}{\epsilon}}\left(\arctan\left(\sqrt{\frac{1}{\epsilon}}(\mcH(x)-c)\right) - \arctan\left(\sqrt{\frac{1}{\epsilon}}(\mcH(y)-c)\right)\right)\right), & \mbox{if } \mcH(y) > \mcH(x) > c. 
		\end{cases}
\end{align*}}

\subsection{Square root $f$: Metropolis-Hastings with sum of square root and logarithmic Hamiltonian}\label{subsec:squarerf}

In the final example, we let $f(z) = \sqrt{z}$ for $z \geq 0$. 
For $x,y \in \{\mcH(y) > \mcH(x) \geq c\}$, consider the integral
$$\int_{\mcH(x)}^{\mcH(y)} \frac{1}{\sqrt{u-c} + \epsilon}\,du = 2\left(\left(\sqrt{\mcH(y)-c} - \sqrt{\mcH(x)-c}\right) - \epsilon \ln \left(\dfrac{\sqrt{\mcH(y)-c} + \epsilon}{\sqrt{\mcH(x)-c} + \epsilon}\right)\right),$$
putting the expression back into \eqref{eq:arint} gives
{\footnotesize
	\begin{align*}
		\exp\left(-(\mcH_{\epsilon,c}^f(y)-\mcH_{\epsilon,c}^f(x))_+\right) &= \begin{cases} 1, &\mbox{if } \mcH(y) \leq \mcH(x); \\
			\exp\left(-\frac{1}{\epsilon}(\mcH(y) - \mcH(x))\right), & \mbox{if } c\geq \mcH(y) > \mcH(x); \\
			\exp\left(-\frac{1}{\epsilon}(c - \mcH(x)) - 2\sqrt{\mcH(y)-c}\right) \left(\dfrac{\sqrt{\mcH(y)-c} + \epsilon}{ \epsilon}\right)^{2\epsilon}, & \mbox{if } \mcH(y) > c \geq \mcH(x); \\
			\exp\left(2\sqrt{\mcH(x)-c} - 2\sqrt{\mcH(y)-c}\right) \left(\dfrac{\sqrt{\mcH(y)-c} + \epsilon}{\sqrt{\mcH(x)-c} + \epsilon}\right)^{2\epsilon}, & \mbox{if } \mcH(y) > \mcH(x) > c. 
		\end{cases}
\end{align*}}

\begin{rk}[Use of landscape modification for sampling]
	This paper focuses on investigating the acceleration effect of landscape modification in stochastic optimization and simulated annealing. Nonetheless, the technique of landscape modification can also be applied to sampling from multimodal distributions. In \cite{JM21}, we analyze the use of landscape modification for sampling. We now briefly describe the setting therein. Suppose that we are interested in sampling from a multimodal distribution that we denote by $\nu(x) \propto e^{-\mcH(x)}$. Applying the idea of landscape modification, we then construct a Metropolis-Hastings chain with the following transformed Hamiltonian function
	\begin{align}\label{eq:mcHalpha}
		\mcH^f_{1,c,\alpha}(x) = \int_{\mcH_{\textrm{min}}}^{\mathcal{H}(x)} \dfrac{1}{\alpha f((u-c)_+) + 1}\,du
	\end{align}
	so that its stationary distribution is given by $\nu^f_{1,c,\alpha}(x) \propto e^{-\mcH^f_{1,c,\alpha}(x)}$. We note that the parameter $\alpha \geq 0$ is introduced in \eqref{eq:mcHalpha}, which controls the bias between the distribution $\nu$ and its landscape-modified counterpart $\nu^f_{1,c,\alpha}$. If we anneal this parameter by sending $\alpha_t \to 0$ as $t \to \infty$, then $\nu^f_{1,c,\alpha_t}$ converges weakly to the target distribution $\nu$. As a result, we construct a non-homogeneous Metropolis-Hastings chain that converges to $\nu$ in the long run while enjoying the benefits of landscape modification.
\end{rk}

\subsection{A Metropolised Ehrenfest urn with landscape modification}\label{subsec:ehrenfest}

In this section, we discuss a \newline Metropolised Ehrenfest urn model and our exposition follows closely as that in \cite[Section $5.1.2$]{DM94}. Let us first briefly fix the setting. We consider the state space $\mathcal{X} = \{0,1,\ldots,d\}$ with $d \in \mathbb{N}$ and take a linear Hamiltonian $\mathcal{H}(x) = x$, where $\mcH_{\textrm{min}} = 0$. The proposal birth-death chain has generator $Q$ given by $Q(x,x+1) = 1 - x/d$, $Q(x,x-1) = x/d$, $Q(x,x) = -1$ and zero otherwise, and we note that the stationary measure of $Q$ is $\mu(x) \propto 2^{-d} {d \choose x}$. With these choices, the classical Metropolised dynamics is $M^0_{\epsilon}(x,x+1) = Q(x,x+1) e^{-\frac{1}{\epsilon}}$, $M^0_{\epsilon}(x,x-1) = Q(x,x-1)$ and $\pi^0(x) \propto e^{-\frac{1}{\epsilon} x} \mu(x)$. It is shown in \cite[equation $(5.1.1)$]{DM94} that
\begin{align}\label{eq:lambda2up}
	\lambda_2(-M^0_{\epsilon}) \leq \lambda_2(-Q) \dfrac{1}{\sum_{x=0}^d e^{-\frac{1}{\epsilon} x} \mu(x)} = \dfrac{d}{2} \dfrac{2^d}{(1+e^{-\frac{1}{\epsilon}})^d}.
\end{align}
At a fixed temperature $\epsilon$, we note that the upper bound in \eqref{eq:lambda2up} is exponential in $d$.

Now, we consider the landscape modified Metropolis-Hastings with $f(z) = z$ and $c = 1$. With these parameters, we compute
$$\mathcal{H}(x) = \int_0^x \dfrac{1}{(u-1)_+ + \epsilon} \, du = \frac{1}{\epsilon} + \ln\left(\dfrac{(x-1)_+ + \epsilon}{\epsilon}\right).$$
Using \cite[equation $(5.1.1)$]{DM94} leads to
\begin{align}\label{eq:lambda2upland}
	\lambda_2(-M^f_{\epsilon,c=1}) \leq \lambda_2(-Q) \dfrac{1}{\sum_{x=0}^d e^{-\mathcal{H}(x)} \mu(x)} = \dfrac{d}{2} \dfrac{2^d e^{\frac{1}{\epsilon}}}{\frac{\epsilon}{\lfloor d/2 \rfloor - 1 + \epsilon} {d \choose {\lfloor d/2 \rfloor}}} \sim d^3 \frac{1}{\epsilon} e^{\frac{1}{\epsilon}},
\end{align}
where we use the Stirling's formula that gives for large enough $d$, 
\begin{align*}
	{d \choose {\lfloor d/2 \rfloor}} &\sim \dfrac{2^d}{\sqrt{\pi \lfloor d/2 \rfloor}}.
\end{align*}
As a result, the upper bound in \eqref{eq:lambda2upland} gives a polynomial dependency on $d$, at the tradeoff of an exponential dependency on $\frac{1}{\epsilon}$. In retrospect this result is not surprising, since we are working with a logarithmic Hamiltonian and hence the asymptotics is in the polynomial of $d$ instead of $2^d$. This section serves as a warm-up example before we discuss the more complicated Curie-Weiss model in Section \ref{sec:CW}.

\subsection{Tuning strategies of $f$ and $c$}\label{subsec:tuning}

In this section, we discuss tuning strategies of $f$ and the threshold parameter $c$ in the context of using the Metropolis-Hastings chain with landscape modification for stochastic optimization.

First, the function $f$ controls how the landscape is transformed above the threshold parameter $c$. For example, taking a linear $f(z) = z$ gives a logarithmic transformation while taking a quadratic $f(z) = z^2$ yields an $\arctan$ transformation as shown in Section \ref{subsec:linearf} and Section \ref{subsec:quadf}. In general, we recommend using $f(z) = z^2$ since $\arctan$ transformation is uniformly bounded above by $\pi/2$, which facilitates exploration on the part of landscape above $c$ by giving a higher transition rate compared with the choice of $f(z) = z$. However, we suspect that in numerical investigations and depending on the target Hamiltonian $\mcH$ there are possibilities of using a linear $f$ to yield improved convergence towards the global minimum.

Second, as we shall see in Section \ref{sec:discreteSA}, the threshold parameter $c$ controls the clipped critical height $c^*$, and ideally it should be set as close to $\mcH_{\textrm{min}}$ as possible. This is possible if we have information about the value $\mcH_{\textrm{min}}$, which is the case for some statistical physics and theoretical computer science models \cite{NZ19,Z18}. In general however, we may not have access to the value $\mcH_{\textrm{min}}$, and one general method is to tune the threshold parameter adaptively by setting the value at time $t$ to be the running minimum up to time $t$ generated by the chain. However, the resulting process becomes non-Markovian because of the adaptive tuning. Another adaptive tuning strategy is to set the value of $c$ at time $t$ to be $c_t = \mcH(y_t) - d$, where $y_t$ is the proposed state at time $t$ generated by the proposal chain and $d \geq 0$ is a fixed number. In this way, the transition rate of the landscape modified MH chain is always greater than or equal to the MH chain without landscape modification. We shall numerically investigate this tuning strategy and report positive results in Section \ref{subsec:numerical}.

As far as this paper is concerned, we assume a fixed $c$ in all of the main theoretical results. We shall postpone to future work a systematic numerical study of investigating various choices of $f$ and tuning strategies of $c$ on benchmark functions, as well as a theoretical analysis of adaptively tuning $c$ by the running minimum generated by the algorithm in Metropolis-Hastings chain with landscape modification.

\subsection{Connections between landscape modification and other acceleration techniques}\label{subsec:connection}

In this section, we outline similarities and differences in idea between Metropolis-Hastings (MH) with landscape modification and other common acceleration techniques in the literature for MH and simulated annealing.

\subsubsection{Catoni's energy transformation algorithm.}

Let $\alpha_1 \geq 0,\alpha_2 \geq 0,\alpha_3 > - \mcH_{\textrm{min}}$ be three parameters. In \cite{Catoni98, Catoni96}, the author introduces the energy transformation algorithm by transforming the Hamiltonian $\mcH$ to
$$F_{\alpha_1,\alpha_2,\alpha_3}(x) := \alpha_1 \mcH(x) + \alpha_2 \log \left(\mcH(x) + \alpha_3 \right).$$
Recall that in Section \ref{subsec:linearf}, MH with landscape modification can be considered as a\textbf{ state-dependent} version of energy transformation if we take $f(z) = z$, $\alpha_3 = -c + \epsilon$ and $\alpha_1$, $\alpha_2$ are chosen in a state-dependent manner:
$$\alpha_1(x) = \mathbf{1}_{\{\mcH(x) \leq c\}}, \quad \alpha_2(x) = \mathbf{1}_{\{\mcH(x) > c\}}.$$

Note that MH with landscape modification can give rise to other kinds of energy transformation by different choices of $f$, see for example the quadratic case or the square root case in Section \ref{subsec:quadf} and \ref{subsec:squarerf} respectively. The idea of mapping or transforming the function from $\mcH$ to $F(\mcH)$ with $F$ being strictly increasing and concave can be dated back to R. Azencott.

\subsubsection{Preconditioning of the Hamiltonian.}

Landscape modification can be understood as a \textbf{state-dependent} preconditioning of the Hamiltonian $\mcH$. Recall that in \eqref{eq:arint} we compute the acceptance-rejection probability in MH by
\begin{align*}
	\exp\left(-(\mcH_{\epsilon,c}^f(y)-\mcH_{\epsilon,c}^f(x))_+\right) &= \begin{cases} 1, &\mbox{if } \mcH(y) \leq \mcH(x); \\
		\exp\left(-\frac{1}{\epsilon}(\mcH(y) - \mcH(x))\right), & \mbox{if } c\geq \mcH(y) > \mcH(x); \\
		\exp\left(-\frac{1}{\epsilon}(c - \mcH(x)) - \int_{c}^{\mcH(y)} \frac{1}{f(u-c) + \epsilon}\,du\right), & \mbox{if } \mcH(y) > c \geq \mcH(x); \\
		\exp\left(-\int_{\mcH(x)}^{\mcH(y)} \frac{1}{f(u-c) + \epsilon}\,du\right), & \mbox{if } \mcH(y) > \mcH(x) > c. 
	\end{cases}
\end{align*}
On the set $\{c\geq \mcH(y) > \mcH(x)\}$, the acceptance-rejection probability is the same as the original MH to allow for exploitation, while on the set $\{\mcH(y) > c \geq \mcH(x)\}$ and $\{\mcH(y) > \mcH(x) > c\}$, the acceptance-rejection probability is higher than that of the original MH to encourage exploration of the landscape. Therefore, there is a higher transition rate of moving to other states when the algorithm is above the threshold $c$.

\subsubsection{Importance sampling.}

In importance sampling, the target distribution is altered for possible benefits and speedups such as variance reduction. In landscape modification, the target distribution in MH is altered from the original Gibbs distribution $\pi^0(x) \propto e^{-\frac{1}{\epsilon} \mcH(x)}$ to $\pi^f_{\epsilon,c}(x) \propto e^{-\mcH_{\epsilon,c}^f(x)}$, while the set of stationary points is preserved in the sense that $\mcH$ and $\mcH^f_{\epsilon,c}$ share the same set of stationary points. In importance sampling however, the set of stationary points need not be preserved between the altered Hamiltonian and the original Hamiltonian.

\subsubsection{Quantum annealing.}

In quantum annealing \cite{WWZ16}, given a target Hamiltonian $\mcH$ and an initial Hamiltonian $\mcH_{init}$ that is usually easy to optimize, we optimize a time-dependent function $Q_t$ defined by, for $t \in [0,T]$,
$$Q_t(x) := A(t) \mcH_{\textrm{init}}(x) + B(t) \mcH(x),$$
where $A(t)$ and $B(t)$ are smooth annealing schedules that satisfy $A(T) = B(0) = 0$ and $T$ is the total annealing time. We also choose $A(t)$ to be decreasing and $B(t)$ to be increasing on the interval $[0,T]$.

In simulated annealing with landscape modification, we also optimize a time-dependent function $\mcH^f_{\epsilon_t,c}$ which shares the same set of stationary points as the target $\mcH$. In quantum annealing, $\mcH_{\textrm{init}}$ and $\mcH$ do not necessarily share the same set of stationary points. We mention the work \cite{DelMiclo99,Lowe96,FG93} for simulated annealing with time-dependent energy function.

\section{The Curie-Weiss model with landscape modification}\label{sec:CW}

In this section, we demonstrate the power of landscape modification by revisiting the Curie-Weiss (CW) model. With appropriate choice of parameters, the landscape of the CW free energy is modified and the local minimum is eliminated while the global minimum is preserved on the transformed function. As a result, landscape modification convexifies the free energy from a double-well to a single-well as a function of the magnetization.

Let us first recall the setting of the CW model of a ferromagnet with external field $h \in \mathbb{R}$ and fix a few notations. We shall follow the setting as in \cite[Chapter $13, 14$]{BH15}. Let $\mathcal{X} = \{-1,1\}^N$ be the set of possible configurations of the CW model with $N \in \mathbb{N}$. The CW Hamiltonian is given by, for $\sigma = (\sigma_i)_{i=1}^N \in \mathcal{X}$,
$$H_{N}(\sigma) := -\dfrac{1}{2N} \sum_{i,j=1}^{N} \sigma_i \sigma_j - h \sum_{i=1}^N \sigma_i = - \dfrac{N}{2} m_N(\sigma)^2 - h N m_N(\sigma) =: N E\left(m_N(\sigma)\right),$$
where $m_N(\sigma) = (1/N)\sum_{i=1}^N \sigma_i$ is the empirical magnetization. Consider the continuized Glauber dynamics by picking a node uniformly at random and flipping the sign of the selected spin, while targeting the Gibbs distribution with the CW Hamiltonian $H_N$ at temperature $\epsilon$. The resulting Metropolis dynamics is given by
$$P_{\epsilon,N}(\sigma,\sigma^{\prime}) = \begin{cases} (1/N) e^{-\frac{1}{\epsilon}(H_N(\sigma^{\prime})-H_N(\sigma))_+}, &\mbox{if } \norm{\sigma-\sigma^{\prime}}_1 = 2; \\
	- \sum_{\eta: \eta \neq \sigma} P_{\epsilon,N}(\sigma,\eta), & \mbox{if } \sigma = \sigma^{\prime}; \\
	0, & \mbox{otherwise,} \end{cases}$$
where $\norm{\cdot}_1$ is the $l^1$ norm on $\mathcal{X}$. 

The dynamics of the empirical magnetization $(m^0(t))_{t \geq 0}$ can be described by lumping the Glauber dynamics to give
\begin{align*}
	M^0_N(m,m^{\prime}) = M^0_{\epsilon,N}(m,m^{\prime}) = \begin{cases} \frac{1-m}{2} e^{-\frac{1}{\epsilon} N(E(m^{\prime})-E(m))_+}, &\mbox{if } m^{\prime} = m + 2N^{-1}; \\
		\frac{1+m}{2} e^{-\frac{1}{\epsilon} N(E(m^{\prime})-E(m))_+}, &\mbox{if } m^{\prime} = m - 2N^{-1}; \\
		- \sum_{m^{\prime}: m^{\prime} \neq m} M^0_{\epsilon,N}(m,m^{\prime}), & \mbox{if } m = m^{\prime}; \\
		0, &\mbox{otherwise.} \end{cases}
\end{align*}
on the state space $\Gamma_N := \{-1,-1+2N^{-1},\ldots,1-2N^{-1},1\}$ with the image Gibbs distribution 
$$\pi^0_{N}(m) = \pi^0_{\epsilon,N}(m) \propto e^{-\frac{1}{\epsilon} N E(m)} {N \choose \frac{1+m}{2}N} 2^{-N}, \quad m \in \Gamma_N,$$
as the stationary distribution. Note that the dependency on $\epsilon$ is suppressed in the notations of $M^0_N$ and $\pi^0_N$. Denote by 
\begin{align*}
	I_N(m) &:= - \dfrac{1}{N} \ln \left({N \choose \frac{1+m}{2}N} 2^{-N}\right), \\
	I(m) &:= \dfrac{1}{2} (1+m) \ln(1+m) + \dfrac{1}{2} (1-m) \ln(1-m), \\
	g_{\epsilon,N}(m) &:= E(m) + \epsilon I_N(m),
\end{align*}
where $I(m)$ is the Cram\'{e}r rate function for coin tossing. As a result, the image Gibbs distribution can be written as 
\begin{align*}
	\pi^0_{\epsilon,N}(m) &\propto e^{-\frac{1}{\epsilon} N g_{\epsilon,N}(m)}, \\
	\lim_{N \to \infty} I_N(m) &= I(m), \\
	g_{\epsilon}(m) &:= \lim_{N \to \infty} g_{\epsilon,N}(m) = E(m) + \epsilon I(m).
\end{align*}
$g_{\epsilon}$ is called the free energy of the CW model. The stationary point(s) of $g_{\epsilon}$ satisfies the classical mean-field equation
\begin{align}\label{eq:mf}
	m = \tanh\left(\frac{1}{\epsilon}(m+h)\right).
\end{align}

To seek the ground state(s) of the free energy, we consider modifying the landscape of the CW Hamiltonian from $H_N$ to
\begin{align}
	E^f_{\epsilon,c}(m) &:= \int_{d}^{E(m)} \dfrac{1}{f((u-c)_+) + \epsilon} \, du, \\
	H^f_{\epsilon,c,N}(\sigma) &:= N \cdot E^f_{\epsilon,c}(m(\sigma)),
\end{align}
where $d \in \mathbb{R}$ can be chosen arbitrarily since we are only interested in the difference of $E^f_{\epsilon,c}$. The infinitesimal generator of the magnetization $(m^f(t))_{t \geq 0}$ is 
\begin{align*}
	M^f_{\epsilon,c,N}(m,m^{\prime}) = \begin{cases} \frac{1-m}{2} e^{- N(E^f_{\epsilon,c}(m^{\prime})-E^f_{\epsilon,c}(m))_+}, &\mbox{if } m^{\prime} = m + 2N^{-1}; \\
		\frac{1+m}{2} e^{- N(E^f_{\epsilon,c}(m^{\prime})-E^f_{\epsilon,c}(m))_+}, &\mbox{if } m^{\prime} = m - 2N^{-1}; \\
		- \sum_{m^{\prime}: m^{\prime} \neq m} M^f_{\epsilon,c,N}(m,m^{\prime}), & \mbox{if } m = m^{\prime}; \\
		0, &\mbox{otherwise,} \end{cases}
\end{align*}
with stationary distribution
$$\pi^f_{\epsilon,c,N}(m) \propto e^{- N E^f_{\epsilon,c,N}(m)} {N \choose \frac{1+m}{2}N} 2^{-N} = e^{- N g^f_{\epsilon,c,N}(m)}, \quad m \in \Gamma_N,$$
where
\begin{align*}
	g^f_{\epsilon,c,N}(m) &:= E^f_{\epsilon,c}(m) +  I_N(m).
\end{align*}
By taking the limit $N \to \infty$, the free energy in the landscape-modified CW model is therefore
$$g^f_{\epsilon,c}(m) := E^f_{\epsilon,c}(m) +  I(m).$$
Setting the derivative of $g^f_{\epsilon,c}$ equals to zero gives the \textbf{landscape-modified mean-field equation}:
\begin{align}\label{eq:mfland}
	m &= \tanh\left(\dfrac{m+h}{f((E(m)-c)_+) + \epsilon}\right).
\end{align}
Observe that if we take $f = 0$, then \eqref{eq:mfland} reduces to the classical mean-field equation in \eqref{eq:mf}.

\subsection{Main results}\label{subsec:CWfix}

Without loss of generality, assume the external magnetic field is $h < 0$. In the subcritical regime where $\frac{1}{\epsilon} > 1$, it is known that there are two local minima of $g_{\epsilon}$. We denote the global minimum of $g_{\epsilon}$ by $m_{-}^* < 0$ and the other local minimum by $m_+^* > 0$, where $|m_-^*| > m_+^*$, and let $z^*$ be the saddle point between $m_-^*$ and $m_+^*$. We also write $m_-^*(N)$ (resp.~ $m_+^*(N)$) to be the closest point in Euclidean distance on $\Gamma_N$ to $m_-^*$ (resp.~ $m_+^*$).

\begin{theorem}[Landscape modification in the subcritical regime]\label{thm:landmodCW}
	Suppose $\frac{1}{\epsilon} > 1$, $h < 0$ and $f,c$ are chosen as in Assumption \ref{assump:main}.
	\begin{enumerate}
		\item\label{it:convexify}[Convexification of the free energy $g_{\epsilon}$ and subexponential mean crossover time] If we choose $c \in [E(m_-^*), E(m_+^*))$, $c < h^2/2$, $-h - \sqrt{h^2 - 2c} \leq z^*$ and for $m \in [-h - \sqrt{h^2 - 2c},-h + \sqrt{h^2 - 2c}]$,
		$$m > \tanh\left(\dfrac{m+h}{f((E(m)-c)_+) + \epsilon}\right),$$
		then $m_-^*$ is the only stationary point of the modified free energy $g^f_{\epsilon,c}$, which is a global minimum. \textcolor{black}{Consequently, we have subexponential mean crossover time on the modified landscape
		\begin{align}\label{eq:convexify}
			\lim_{N \to \infty} \dfrac{1}{N} \log \E_{m_+^*(N)} \left(\tau^f_{m_-^*(N)}\right) = 0,
		\end{align}
		while the mean crossover time on the original landscape is exponential in $N, 1/\epsilon$ and the original critical height $g_{\epsilon}(z^*) - g_{\epsilon}(m_+^*)$ with
		\begin{align*}
			\lim_{N \to \infty} \dfrac{1}{N} \log \E_{m_+^*(N)} \left(\tau^0_{m_-^*(N)}\right) = \frac{1}{\epsilon} (g_{\epsilon}(z^*) - g_{\epsilon}(m_+^*)).
		\end{align*}}
		
		\item\label{it:clip} If we choose $c \in [E(m_+^*), E(z^*)]$ and assume in addition that $f$ is twice differentiable and satisfies $f^{\prime}(0) = f^{\prime \prime}(0) = 0$, then there exists $\mathbf{z}^* = \arg \max_{m_-^* \leq m \leq m_+^*} g^f_{\epsilon,c}(m)$ and as $N \to \infty$,
		\begin{align*}
			\E_{m_+^*(N)} \left(\tau^f_{m_-^*(N)}\right) &= \exp\left(N (g_{\epsilon,c}^f(\mathbf{z}^{*}) - g_{\epsilon,c}^f(m_+^*))\right) \\
			&\quad \times \frac{2}{1-|\mathbf{z}^{*}|} \sqrt{\frac{1-\mathbf{z}^{*2}}{1-m_{+}^{* 2}}} \frac{2 \pi N / 4}{ \sqrt{\left(-(g_{\epsilon,c}^{f})^{\prime \prime}\left(\mathbf{z}^{*}\right)\right) (g_{\epsilon,c}^{f})^{\prime \prime}\left(m_{+}^{*}\right)}} (1+o(1)).
		\end{align*}
		Consequently,
		\begin{align*}
			\lim_{N \to \infty} \dfrac{1}{N} \log \E_{m_+^*(N)} \left(\tau^f_{m_-^*(N)}\right) &= g_{\epsilon,c}^f(\mathbf{z}^{*}) - g_{\epsilon,c}^f(m_+^*) \\
			&= \frac{1}{\epsilon} (c - E(m_+^*)) + \int_c^{E(\mathbf{z}^{*})} \dfrac{1}{f((u-c)_+) + \epsilon} \,du+ \left(I(\mathbf{z}^{*}) - I(m_+^*)\right) \\
			&\leq \frac{1}{\epsilon} (g_{\epsilon}(z^*) - g_{\epsilon}(m_+^*)) = \lim_{N \to \infty} \dfrac{1}{N} \log \E_{m_+^*(N)} \left(\tau^0_{m_-^*(N)}\right).
		\end{align*}
	\end{enumerate}
\end{theorem}

Before we present the proof, we interpret the results in Theorem \ref{thm:landmodCW} intuitively: in item \eqref{it:convexify}, on the one hand we would like to choose $c$ small enough such that the mapping $$m \mapsto \tanh\left(\dfrac{m+h}{f((E(m)-c)_+) + \epsilon}\right)$$ is flattened and only intersects with the straight line $m \mapsto m$ at the global minimum $m_-^*$. In this way the landscape of $g^f_{\epsilon,c}$ is transformed from a double-well to a single-well, while the location of the global minimum at $m_-^*$ is preserved as that in the original landscape $g_{\epsilon}$. This is illustrated in Figure \ref{fig:CWfreeenergy} and Figure \ref{fig:CWmf}. On the other hand, we cannot choose $c$ to be too small if we are interested in seeking the ground state of $g_{\epsilon}$, since otherwise if $c < E(m_-^*)$ then $m_-^*$ may no longer be the global minimum in the transformed free energy $g^f_{\epsilon,c}$. \textcolor{black}{This consequently yields a subexponential in $N$ mean crossover time on the modified landscape, while the original mean crossover time is exponential in $N, 1/\epsilon$ and the original critical height $g_{\epsilon}(z^*) - g_{\epsilon}(m_+^*)$}. In Theorem \ref{thm:landmodCW} item \eqref{it:clip}, we choose a larger value of $c$ compared with that in item \eqref{it:convexify}. Although the transformed free energy $g^f_{\epsilon,c}$ is not a convex function, it has a smaller critical height than the original free energy $g_{\epsilon}$. \textcolor{black}{This subsequently gives a reduced exponential dependence on the modified mean crossover time compared with the original mean crossover time.}

The power of landscape modification or energy transformation lies in tuning the parameter $c$ appropriately. One way to tune $c$ is to use the running minimum generated by the algorithm on the original free energy $g_{\epsilon}$. Suppose we start in the well containing the local minimum $m_+^*$, and setting $c$ to be the running minimum eventually gives $c = E(m_+^*)$, and hence Theorem \ref{thm:landmodCW} item \eqref{it:clip} can be applied and the critical height on the modified landscape is reduced.

We illustrate Theorem \ref{thm:landmodCW} with a concrete numerical example in Figure \ref{fig:CWfreeenergy} and Figure \ref{fig:CWmf}, where we take $h = -0.05$ and $f(z) = z$ at temperature $\epsilon = 1/1.5$. We numerically compute that $m_-^* = -0.8863$, $m_+^* = 0.8188$ and $z^* = 0.1524$. As a result we have $E(m_-^*) = -0.4371$, $E(m_+^*) = -0.2943$ and $E(z^*) = -0.004$. In the leftmost plot of Figure \ref{fig:CWfreeenergy} and Figure \ref{fig:CWmf}, we choose $c = -0.4 \in [E(m_-^*), E(m_+^*))$. We numerically check that the conditions in Theorem \ref{thm:landmodCW} item \eqref{it:convexify} are satisfied, and we see that the blue curve and the orange curve share the same location of the global minimum. In the rightmost plot of Figure \ref{fig:CWfreeenergy} and Figure \ref{fig:CWmf}, we choose $c = -0.2 \in [E(m_+^*), E(z^*)]$. We see that the blue curve and the red curve share the same locations of the two local minima, while the critical height is smaller than that in the original landscape $g_{\epsilon}$.

\begin{figure}[h]
	\includegraphics[width=1.1\textwidth,center]{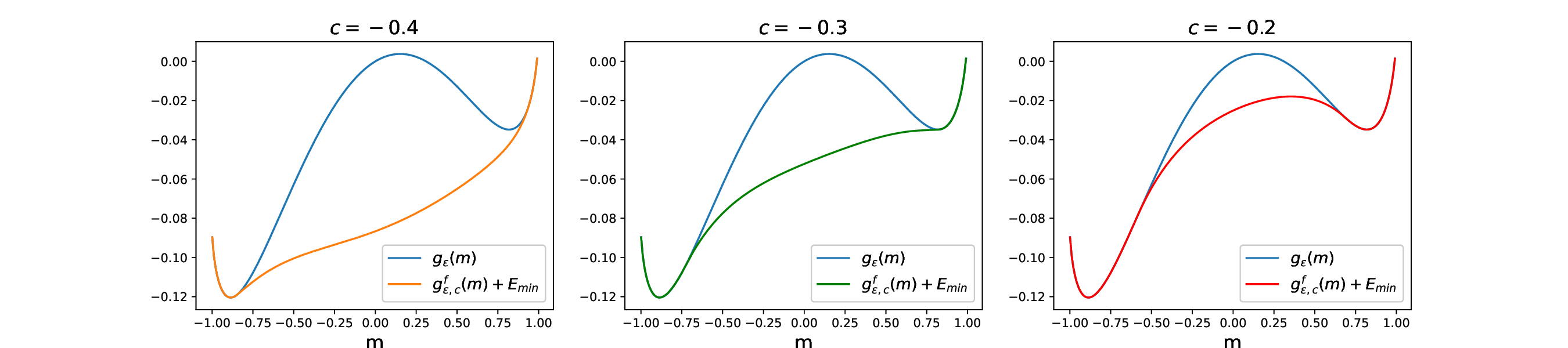} 
	\caption{Plots of the free energy $g_{\epsilon}$ and the modified free energy $g^{f}_{\epsilon,c} + E_{min}$ with $h = -0.05$ and $f(z) = z$ at temperature $\epsilon = 1/1.5$, where $E_{min} = \min_{m \in [-1,1]} E(m)$. We shift the modified free energy by $E_{min}$ so that it is on the same scale as the original free energy $g_{\epsilon}$.}
	\label{fig:CWfreeenergy} 
\end{figure}
\newpage
\begin{figure}[h]
	\centering
	\includegraphics[width=1.1\textwidth,center]{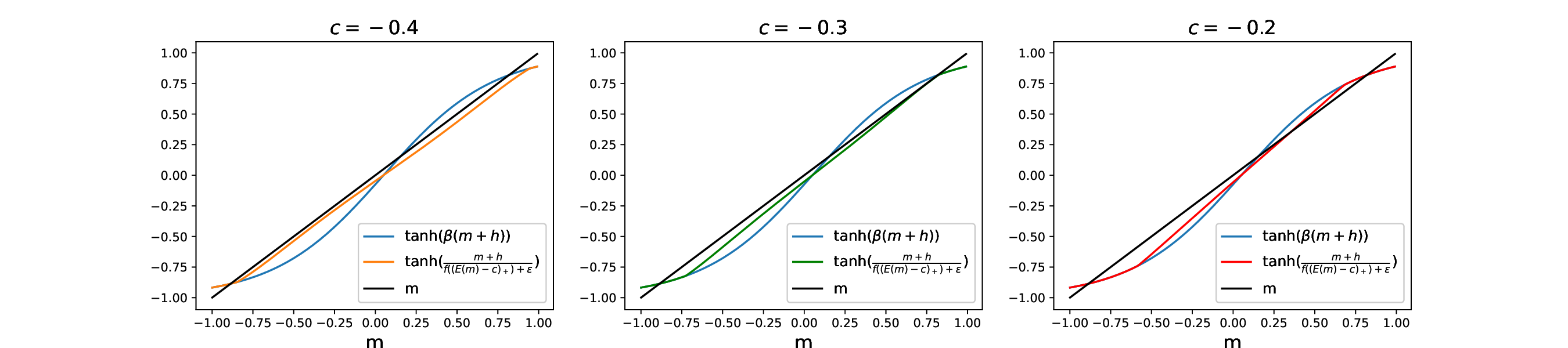} 
	\caption{Plots of the mean-field equation \eqref{eq:mf} and the modified mean-field equation \eqref{eq:mfland} with $h = -0.05$ and $f(z) = z$ at temperature $\epsilon = 1/1.5$.}
	\label{fig:CWmf} 
\end{figure}

\subsection{Proof of Theorem \ref{thm:landmodCW}}

%
Before we give the proof, let us first recall the concept of critical height and the notation $L$ as introduced in Section \ref{sec:MHland}. In this Section, we are interested in the CW model with and without landscape modification we with free energy $g^f_{\epsilon,c,N}$ and $g_{\epsilon,N}$ respectively. As a result we define the analogous concepts of critical heights by inserting a subscript of $N$. This leads us to
\begin{align*}
	H^f_{\epsilon,c,N} &:= L(g^f_{\epsilon,c,N}) = \max_{x,y \in \mathcal{X}}\{G(g^f_{\epsilon,c,N},x,y) - g^f_{\epsilon,c,N}(x) - g^f_{\epsilon,c,N}(y) \} + \min_{x \in \Gamma_N} g^f_{\epsilon,c,N}(x), \\
	H^0_{\epsilon,N} &:= L(g_{\epsilon,N}) = \max_{x,y \in \mathcal{X}}\{G(g_{\epsilon,N},x,y) - g^f_{\epsilon,N}(x) - g_{\epsilon,N}(y) \} + \min_{x \in \Gamma_N} g_{\epsilon,N}(x),
\end{align*}
where we recall that $G$ is introduced in \eqref{eq:G}.

We proceed with the 
\begin{proof}[Proof of Theorem \ref{thm:landmodCW}]
	First, we prove item \eqref{it:convexify}. We observe that $\{E(m) \geq c\} = \{m \in [-h - \sqrt{h^2 - 2c},-h + \sqrt{h^2 - 2c}]\}$. On this interval,
	$$\dfrac{d}{dm} g^f_{\epsilon,c}(m) = \arctanh(m) - \dfrac{m+h}{f((E(m)-c)_+) + \epsilon} > 0,$$
	and hence the modified free energy is strictly increasing on this interval. On the interval $\{m > -h + \sqrt{h^2 - 2c}\}$, $\dfrac{d}{dm} g^f_{\epsilon,c}(m) = \dfrac{d}{dm} g_{\epsilon}(m) > 0$ as the original free energy is strictly increasing. On the interval $\{m < -h - \sqrt{h^2 - 2c}\}$, we also have $\dfrac{d}{dm} g^f_{\epsilon,c}(m) = \dfrac{d}{dm} g_{\epsilon}(m)$. Thus, with these parameter choices, the only stationary point of $g^f_{\epsilon,c}$ is $m^*_-$, which is the global minimum.
	
	Next, we proceed to prove \eqref{eq:convexify}. According to \cite[Theorem $2.1$]{Lowe96}, for $\xi_1(N)$ a polynomial function in $N$, we have
	$$\dfrac{1}{\lambda_2(-M^f_{\epsilon,c,N})} \leq \xi_1(N) e^{N H^f_{\epsilon,c,N}}.$$
	Now, using the random target lemma \cite[Section $4.2$]{AF14} and the above inequality lead to
	\begin{align}\label{eq:limsuphit}
		\pi^f_{\epsilon,c,N}(m^*_-(N))  \E_{m_+^*(N)} \left(\tau^f_{m_-^*(N)}\right) \leq \sum_{y \in \Gamma_N} \pi^f_{\epsilon,c,N}(y)  \E_{m_+^*(N)} \left(\tau^f_{y}\right) \leq (|\Gamma_N|-1) \dfrac{1}{\lambda_2(-M^f_{\epsilon,c,N})} \leq N \xi_1(N) e^{N H^f_{\epsilon,c,N}}. 
	\end{align}
	Now, let $\overline{m}(N) := \arg \min g^f_{\epsilon,c,N}(m)$ and compute
	\begin{align*}
		g^f_{\epsilon,c,N}(m^*_-(N)) - g^f_{\epsilon,c,N}(\overline{m}(N)) &= E^f_{\epsilon,c}(m^*_-(N)) - E^f_{\epsilon,c}(\overline{m}(N)) + I_N(m^*_-(N)) - I_N(\overline{m}(N)) \\
		&= E^f_{\epsilon,c}(m^*_-(N)) - E^f_{\epsilon,c}(\overline{m}(N)) \\
		&\quad + I_N(m^*_-(N)) - I(m^*_-(N)) + I(m^*_-(N)) - I_N(\overline{m}(N))  \\
		&= E^f_{\epsilon,c}(m^*_-(N)) - E^f_{\epsilon,c}(\overline{m}(N)) \\
		&\quad + [1+o(1)] \frac{1}{2 N} \ln \left(\frac{\pi N\left(1-m^*_-(N)^{2}\right)}{2}\right) \\
		&\quad + I(m_-^*(N)) - I(\overline{m}(N)) \\
		&\quad - [1+o(1)] \frac{1}{2 N} \ln \left(\frac{\pi N\left(1-\overline{m}(N)^{2}\right)}{2}\right) \\
		&\rightarrow 0 \quad \text{as} \, N \to \infty,
	\end{align*}
	where we use \cite[equation $(13.2.5)$]{BH15} in the third equality, and $m_-^*(N), \overline{m}(N) \to m_-^*$ as $N \to \infty$. The above computation combined with \eqref{eq:limsuphit} yields
	\begin{align}\label{eq:limsuphit2}
		\limsup_{N \to \infty} \E_{m_+^*(N)} \left(\tau^f_{m_-^*(N)}\right) \leq \lim_{N \to \infty}  \dfrac{1}{N} \log H^f_{\epsilon,c,N} = 0.
	\end{align}
	
	On the other hand, as the magnetization $(m^f(t))_{t \geq 0}$ is a birth-death process, using \cite[equation $(13.2.2)$]{BH15} the mean hitting time can be calculated explicitly as
	\begin{align*}
		\mathbb{E}_{m_{+}^{*}(N)}\left(\tau^f_{m_{-}^{*}(N)}\right) &= \sum_{m, m^{\prime} \in \Gamma_{N}, m \leq m^{\prime} \atop m_{-}^{*}(N)<m \leq m_{+}^{*}(N)} \frac{\pi^f_{\epsilon,c, N}\left(m^{\prime}\right)}{\pi^f_{\epsilon,c, N}(m)} \frac{1}{M^f_{\epsilon,c, N}\left(m, m-2 N^{-1}\right)} \\
		&\geq \dfrac{1}{M^f_{\epsilon,c, N}\left(m_{+}^{*}(N), m_{+}^{*}(N)-2 N^{-1}\right)} = \dfrac{2}{1+m_{+}^{*}(N)} e^{ N(E^f_{\epsilon,c}(m_{+}^{*}(N)-2 N^{-1})-E^f_{\epsilon,c}(m))_+} \\
		&\geq \dfrac{2}{1+m_{+}^{*}(N)}.
	\end{align*}
	As a result, as $m_+^*(N) \to m_+^*$ we have
	\begin{align}\label{eq:liminfhit}
		\liminf_{N \to \infty} \dfrac{1}{N} \log \E_{m_+^*(N)} \left(\tau^f_{m_-^*(N)}\right)  \geq 0.
	\end{align}
	Using both \eqref{eq:liminfhit} and \eqref{eq:limsuphit2} gives \eqref{eq:convexify}.
	
	Next, we prove item \eqref{it:clip}, which follow closely with the proof of \cite[Theorem $13.1$]{BH15}. Since we choose $c \in [E(m_+^*), E(z^*)]$, then we have $E(m_-^*) - c < E(m_+^*) - c \leq 0$, and hence the landscape modified mean-field equation \eqref{eq:mfland} has at least two solutions $m_+^*$ and $m_-^*$, which are exactly the same as the original mean-field equation \eqref{eq:mf}. As the landscape modified mean-field equation is continuous in $m$, there exists $\mathbf{z}^* = \arg \max_{m_-^* \leq m \leq m_+^*} g^f_{\epsilon,c}(m)$ which also satisfies \eqref{eq:mfland}. Now, for $m \in \Gamma_N$ we consider
	\begin{align*}
		N \left(E^f_{\epsilon,c}(m - 2N^{-1}) - E^f_{\epsilon,c}(m)\right) &= N \int_{E(m)}^{E(m) + 2N^{-1}(m+h - N^{-1})} \dfrac{1}{f((u-c)_+) + \epsilon} \, du \\
		&\rightarrow \dfrac{2(m+h)}{f((E(m)-c)_+) + \epsilon} \quad \text{as } N \to \infty.
	\end{align*}
	If we take $N \to \infty$ and $m \rightarrow \mathbf{z}^*$, we obtain
	\begin{align*}
		\frac{1}{M^f_{\epsilon,c, N}\left(m, m-2 N^{-1}\right)} \to \dfrac{2}{1+\mathbf{z}^*} \exp \left(\dfrac{2(\mathbf{z}^*+h)_+}{f((E(\mathbf{z}^*)-c)_+) + \epsilon}\right) = \dfrac{2}{1 - |\mathbf{z}^*|},
	\end{align*}
	since if $\mathbf{z}^* > 0$, then $\mathbf{z}^* + h > 0$ and satisfies the mean-field equation \eqref{eq:mfland}. Using the mean hitting time formula again leads to, for any $\delta > 0$,
	\begin{align}\label{eq:meanhit}
		\mathbb{E}_{m_{+}^{*}(N)}\left(\tau^f_{m_{-}^{*}(N)}\right) &= \sum_{m, m^{\prime} \in \Gamma_{N}, m \leq m^{\prime} \atop m_{-}^{*}(N)<m \leq m_{+}^{*}(N)} \frac{\pi^f_{\epsilon,c, N}\left(m^{\prime}\right)}{\pi^f_{\epsilon,c, N}(m)} \frac{1}{M^f_{\epsilon,c, N}\left(m, m-2 N^{-1}\right)} \nonumber \\
		&= e^{N\left[g^f_{\epsilon,c,N}\left(\mathbf{z}^{*}\right) - g^f_{\epsilon,c,N}\left(m_{+}^{*}\right)\right]} \frac{2}{1 - |\mathbf{z}^*|}[1+o(1)] \nonumber \\
		&\quad \times \sum_{m, m^{\prime} \in \Gamma_{N} \atop\left|m-\mathbf{z}^{*}\right|<\delta,\left|m^{\prime}-m_{+}^{*}\right|<\delta} \mathrm{e}^{N\left[g^f_{\epsilon,c, N}(m)-g^f_{\epsilon,c, N}\left(\mathbf{z}^{*}\right)\right]- N\left[g^f_{\epsilon,c, N}\left(m^{\prime}\right)-g^f_{\epsilon,c, N}\left(m_{+}^{*}\right)\right]}.
	\end{align}
	Using the Stirling's formula and the same argument as in \cite[equation $(13.2.7)$]{BH15} yield
	\begin{align*}
		e^{N\left[g^f_{\epsilon,c,N}\left(\mathbf{z}^{*}\right) - g^f_{\epsilon,c,N}\left(m_{+}^{*}\right)\right]} &= (1+o(1)) e^{N\left[g^f_{\epsilon,c}\left(\mathbf{z}^{*}\right) - g^f_{\epsilon,c}\left(m_{+}^{*}\right)\right]} \sqrt{\dfrac{1-\mathbf{z}^{*2}}{1-m_+^{*2}}}
	\end{align*}
	followed by substitution to \eqref{eq:meanhit} gives
	\begin{align}\label{eq:meanhit2}
		\mathbb{E}_{m_{+}^{*}(N)}\left(\tau^f_{m_{-}^{*}(N)}\right) &= e^{N\left[g^f_{\epsilon,c}\left(\mathbf{z}^{*}\right) - g^f_{\epsilon,c}\left(m_{+}^{*}\right)\right]} \frac{2}{1 - |\mathbf{z}^*|} \sqrt{\dfrac{1-\mathbf{z}^{*2}}{1-m_+^{*2}}}[1+o(1)] \nonumber \\
		&\quad \times \sum_{m, m^{\prime} \in \Gamma_{N} \atop\left|m-\mathbf{z}^{*}\right|<\delta,\left|m^{\prime}-m_{+}^{*}\right|<\delta} \sqrt{\dfrac{1-m^{2}}{1-\mathbf{z}^{*2}}} \sqrt{\dfrac{1-m^{*2}_+}{1-m^{\prime 2}}} \mathrm{e}^{N\left[g^f_{\epsilon,c}(m)-g^f_{\epsilon,c}\left(\mathbf{z}^{*}\right)\right]- N\left[g^f_{\epsilon,c}\left(m^{\prime}\right)-g^f_{\epsilon,c}\left(m_{+}^{*}\right)\right]}.
	\end{align}
	We proceed to use a Laplace method argument to handle the sum in \eqref{eq:meanhit2}. Note that as $f$ is assumed to be twice-differentiable with $f(0) = f^{\prime}(0) = f^{\prime \prime}(0) = 0$, this implies $g^f_{\epsilon,c}$ is three-time differentiable, and applying the third-order Taylor expansion gives
	\begin{align}
		g^f_{\epsilon,c}(m) - g^f_{\epsilon,c}\left(\mathbf{z}^{*}\right) &= \dfrac{(m - \mathbf{z}^{*})^2}{2} (g_{\epsilon,c}^{f})^{\prime \prime}(	\mathbf{z}^*) + \mathcal{O}\left((m - \mathbf{z}^{*})^3\right), \\
		g^f_{\epsilon,c}\left(m^{\prime}\right)-g^f_{\epsilon,c}\left(m_{+}^{*}\right) &= \dfrac{(m^{\prime} - m_{+}^{*})^2}{2} (g_{\epsilon,c}^{f})^{\prime \prime}(	m_{+}^{*}) + \mathcal{O}\left((m^{\prime} - m_{+}^{*})^3\right),
	\end{align}
	where we use $(g_{\epsilon,c}^{f})^{\prime}(\mathbf{z}^*) = (g_{\epsilon,c}^{f})^{\prime}(m_{+}^{*}) = 0$. Now, we observe that the sum in \eqref{eq:meanhit2} is
	\begin{align*}
		&\quad (1+o(1)) \dfrac{N}{4} \int_{\mathbb{R}} \int_{\mathbb{R}} \exp \left[\frac{1}{2} (g_{\epsilon,c}^{f})^{\prime \prime}\left(\mathbf{z}^{*}\right) u^{2}-\frac{1}{2}  (g_{\epsilon,c}^{f})^{\prime \prime}\left(m_{+}^{*}\right) u^{\prime 2}\right]\, du du^{\prime} \\
		&= (1+o(1)) \dfrac{N}{4} \frac{2 \pi}{\sqrt{\left[-(g_{\epsilon,c}^{f})^{\prime \prime}\left(\mathbf{z}^{*}\right)\right] (g_{\epsilon,c}^{f})^{\prime \prime}\left(m_{+}^{*}\right)}},
	\end{align*}
	since $(g_{\epsilon,c}^{f})^{\prime \prime}\left(\mathbf{z}^{*}\right) < 0$ and $(g_{\epsilon,c}^{f})^{\prime \prime}\left(m_{+}^{*}\right) > 0$.
\end{proof}

\subsection{Extension to the random field Curie-Weiss}

In Section \ref{subsec:CWfix}, we discuss the classical CW model with landscape modification under \textit{fixed} magnetic field. In this section, we aim at considering the random field CW model with landscape modification, and discuss related metastability results and its ground state free energy in such setting, with the aim to illustrate that landscape modification can also be applied in the setting of random energy landscape. Let us begin by first recalling the random field CW model. We shall adapt the setting as in \cite{MP98}. Let $(h_i)_{i \in \mathbb{N}}$ be a sequence of i.i.d. random variables with $\mathbb{P}(h_i = 1) = \mathbb{P}(h_i = -1) = 1/2$. We consider the random Hamiltonian function given by, for a fixed $\theta > 0$ and $\sigma \in \{-1,1\}^N$,
$$\mathbf{H}_N(\sigma) = \mathbf{H}_N(\sigma, \omega) := - \dfrac{N}{2} \mathbf{m}_N(\sigma)^2 - \theta \sum_{i=1}^N h_i(\omega) \sigma_i = - \dfrac{N}{2} (\mathbf{m}_N^+(\sigma) + \mathbf{m}_N^-(\sigma))^2 - \theta N (\mathbf{m}_N^+(\sigma) - \mathbf{m}_N^-(\sigma)),$$
where $\mathbf{m}_N(\sigma) := (1/N) \sum_{i=1}^N \sigma_i$, $\mathbf{m}_N^+(\sigma) := (1/N) \sum_{i=1;~ h_i = 1}^N \sigma_i$ and $\mathbf{m}_N^-(\sigma) := (1/N) \sum_{i=1;~ h_i = -1}^N \sigma_i$. In the sequel, we shall suppress the dependency on $\omega$. Denote the Gibbs distribution at temperature $\epsilon$ on $\{-1,1\}^N$ by
$$\bm{\nu}_N(\sigma) \propto \exp\bigg\{-\frac{1}{\epsilon} \mathbf{H}_N(\sigma)\bigg\}.$$
Let $N^+ := |\{i;~h_i = +1\}|$, $N^- := |\{i;~h_i = -1\}|$ and define the random set
$$\mathbf{\Gamma}_N := \left(-\frac{N^{+}}{N},-\frac{N^{+}}{N} + \frac{2}{N}, \ldots, \frac{N^{+}}{N}\right) \times\left(-\frac{N^{-}}{N},-\frac{N^{-}}{N} + \frac{2}{N}, \ldots, \frac{N^{-}}{N}\right).$$
For $\mathbf{m} = (\mathbf{m}^+, \mathbf{m}^-) \in \bm{\Gamma}_N$, with slight abuse of notation we write
$$\mathbf{H}_N(\mathbf{m}) = - \dfrac{N}{2} (\mathbf{m}^+ + \mathbf{m}^-)^2- \theta N (\mathbf{m}^+ - \mathbf{m}^-) =: N \cdot \mathbf{E}(\mathbf{m}),$$
where $\mathbf{m}^+$ and $\mathbf{m}^-$ are the magnetization among the sites $i$ where respectively
$h_i = 1$ and $h_i = -1$. Let $\bm{\pi}_{\epsilon,N}^0$ denote the image Gibbs distribution of $\bm{\nu}_N$ by $\mathbf{\Gamma}_N$, where
\begin{align*}
	\bm{\pi}_{\epsilon,N}^0(\mathbf{m}) &\propto \exp\bigg\{- \frac{1}{\epsilon} N  \mathbf{g}_{\epsilon,N}(\mathbf{m})\bigg\}, \\
	\mathbf{g}_{\epsilon,N}(\mathbf{m}) &:= - \dfrac{1}{2} (\mathbf{m}^+ + \mathbf{m}^-)^2- \theta (\mathbf{m}^+ - \mathbf{m}^-) - \frac{1}{\frac{1}{\epsilon} N} \log \left(\begin{array}{c}
		N^{+} \\
		\frac{N^{+}}{2}+\mathbf{m}^{+} \frac{N}{2}
	\end{array}\right)\left(\begin{array}{c}
		N^{-} \\
		\frac{N^{-}}{2}+\mathbf{m}^{-} \frac{N}{2}
	\end{array}\right).
\end{align*}
As $N \to \infty$, by the strong law of large number $\mathbf{g}_{\epsilon,N}$ converges almost surely to the free energy given by
$$\mathbf{g}_{\epsilon}(\mathbf{m}) := - \dfrac{1}{2} (\mathbf{m}^+ + \mathbf{m}^-)^2- \theta (\mathbf{m}^+ - \mathbf{m}^-) + \frac{1}{2\frac{1}{\epsilon}}\left(I(2\mathbf{m}^+) + I(2\mathbf{m}^-)\right),$$
where $I(m)$ is the Cram\'{e}r rate function as introduced in Section \ref{sec:CW}. The critical points of $\mathbf{g}_{\epsilon}$ satisfy
\begin{align}\label{eq:mfRFCW}
	\mathbf{m}^+ &= \dfrac{1}{2} \tanh\left(\frac{1}{\epsilon}(\mathbf{m}^+ + \mathbf{m}^- + \theta)\right), \\
	\mathbf{m}^- &= \dfrac{1}{2} \tanh\left(\frac{1}{\epsilon}(\mathbf{m}^+ + \mathbf{m}^- - \theta)\right). \label{eq:mfRFCW2}
\end{align}
In this section, we shall only consider the subcritical regime where $\frac{1}{\epsilon} > \cosh^2(\frac{1}{\epsilon} \theta)$. It can be shown (see e.g. \cite{MP98}) that there are exactly three critical points. Let $\mathbf{m}_{*} > 0$ be the unique positive solution to the mean-field equation
$$\mathbf{m}_{*} = \dfrac{1}{2}\left(\tanh\left(\frac{1}{\epsilon}(\mathbf{m}_{*} + \theta)\right) + \tanh\left(\frac{1}{\epsilon}(\mathbf{m}_{*} - \theta)\right)\right).$$ 
The three critical points of $\mathbf{g}_{\epsilon}$ are given by 
\begin{align*}
	\mathbf{m}_{0} &= \left(\frac{1}{2} \tanh \left(\frac{1}{\epsilon} \theta\right),-\frac{1}{2} \tanh \left(\frac{1}{\epsilon} \theta\right)\right), \\
	\mathbf{m}_{1} &= \left(\frac{1}{2} \tanh \left(\frac{1}{\epsilon} \mathbf{m}_{*}+\frac{1}{\epsilon} \theta\right), \frac{1}{2} \tanh \left(\frac{1}{\epsilon} \mathbf{m}_{*}-\frac{1}{\epsilon} \theta\right)\right), \\
	\mathbf{m}_{2} &= \left(\frac{1}{2} \tanh \left(-\frac{1}{\epsilon} \mathbf{m}_{*}+\frac{1}{\epsilon} \theta\right),-\frac{1}{2} \tanh \left(\frac{1}{\epsilon} \mathbf{m}_{*}+\frac{1}{\epsilon} \theta\right)\right),
\end{align*}
where $\mathbf{m}_0$ is the saddle point and $\mathbf{m}_1, \mathbf{m}_2$ are the two global minima. Consider the continuized Glauber dynamics $(\sigma_N(t))_{t \geq 0}$ by picking a node uniformly at random and changing the sign of the selected spin, while targeting the Gibbs distribution $\bm{\nu}_N$ at temperature $\epsilon$. Denote by $\mathbf{m}_N(t) := \mathbf{m}_N(\sigma_N(t))$ be the induced dynamics on the magnetization, and its infiniteismal generator by $\mathbf{M}^0_{\epsilon,N}$. This is proven to be a Markov chain in \cite{MP98}, with stationary measure $\bm{\pi}^0_{\epsilon,N}$.

Now, let us consider the landscape modified Hamiltonian on $\bm{\Gamma}_N$:
\begin{align}
	\mathbf{E}^f_{\epsilon,c}(\mathbf{m}) &:= \int_{d}^{\mathbf{E}(\mathbf{m})} \dfrac{1}{f((u-c)_+) + \epsilon} \, du, \\
	\mathbf{H}^f_{\epsilon,c,N}(\mathbf{m}) &:= N \cdot \mathbf{E}^f_{\epsilon,c}(\mathbf{m}),
\end{align}
where $d \in \mathbb{R}$ can be chosen arbitrarily since we are only interested in the difference of $\mathbf{E}^f_{\epsilon,c}$. The transformed image Gibbs distribution is therefore
\begin{align*}
	\bm{\pi}_{\epsilon,c,N}^f(\mathbf{m}) &\propto \exp\{- N \mathbf{g}^f_{\epsilon,c,N}(\mathbf{m})\}, \\
	\mathbf{g}^f_{\epsilon,c,N}(\mathbf{m}) &:= \mathbf{E}^f_{\epsilon,c}(\mathbf{m}) - \frac{1}{N} \log \left(\begin{array}{c}
		N^{+} \\
		\frac{N^{+}}{2}+\mathbf{m}^{+} \frac{N}{2}
	\end{array}\right)\left(\begin{array}{c}
		N^{-} \\
		\frac{N^{-}}{2}+\mathbf{m}^{-} \frac{N}{2}
	\end{array}\right).
\end{align*}
The strong law of large number yields that as $N \to \infty$, $\mathbf{g}^f_{\epsilon,c,N}$ converges almost surely to the transformed free energy
$$\mathbf{g}^f_{\epsilon,c}(\mathbf{m}) := \mathbf{E}^f_{\epsilon,c}(\mathbf{m}) + \frac{1}{2}\left(I(2\mathbf{m}^+) + I(2\mathbf{m}^-)\right).$$
The critical points of $\mathbf{g}^f_{\epsilon,c}$ satisfy the following \textbf{landscape modified mean-field equations}:
\begin{align}
	\mathbf{m}^+ &= \dfrac{1}{2} \tanh\left(\dfrac{\mathbf{m}^+ + \mathbf{m}^- + \theta}{f((\mathbf{E}(\mathbf{m})-c)_+) + \epsilon}\right), \label{eq:mfRFCWland} \\
	\mathbf{m}^- &= \dfrac{1}{2} \tanh\left(\dfrac{\mathbf{m}^+ + \mathbf{m}^- - \theta}{f((\mathbf{E}(\mathbf{m})-c)_+) + \epsilon}\right). \label{eq:mfRFCWland2}
\end{align}
Note that \eqref{eq:mfRFCWland} and \eqref{eq:mfRFCWland2} reduce to the classical case \eqref{eq:mfRFCW} and \eqref{eq:mfRFCW2} if we take $f = 0$. Consider the continuized Glauber dynamics $(\sigma_N^f(t))_{t \geq 0}$ by picking a node uniformly at random and changing the sign of the selected spin, while targeting the Gibbs distribution with Hamiltonian $\epsilon \mathbf{H}^f_{\epsilon,c,N}$ at temperature $\epsilon$. Denote by $\mathbf{m}_N^f(t) := \mathbf{m}^f_N(\sigma_N^f(t))$ be the induced dynamics on the magnetization, and its infiniteismal generator by $\mathbf{M}^f_{\epsilon,c,N}$, which is a Markov chain with stationary measure $\bm{\pi}^f_{\epsilon,c,N}$.

In the following, we shall consider the case where $c \in [\mathbf{E}(\mathbf{m}_1), \mathbf{E}(\mathbf{m}_0)]$. It can be seen that the two global minima of $\mathbf{g}^f_{\epsilon,c}$ remain to be $\mathbf{m}_1, \mathbf{m}_2$ with this choice of $c$. For any path $\gamma^{\mathbf{m}_1,\mathbf{m}_0}$ connecting $\mathbf{m}_1$ and $\mathbf{m}_0$, we define
\begin{align*}
	\mathbf{m}_3(\gamma^{\mathbf{m}_1,\mathbf{m}_0}) &:= \arg \max \{\mathbf{g}^f_{\epsilon,c}(\gamma_i);~\gamma_i \in \gamma^{\mathbf{m}_1,\mathbf{m}_0}\}, \\
	\Delta \mathbf{g}^f_{\epsilon,c} &:= \min_{\gamma^{\mathbf{m}_1,\mathbf{m}_0}} \mathbf{m}_3(\gamma^{\mathbf{m}_1,\mathbf{m}_0}) - \mathbf{g}^f_{\epsilon,c}(\mathbf{m}_1) = \min_{\gamma^{\mathbf{m}_1,\mathbf{m}_0}} \mathbf{m}_3(\gamma^{\mathbf{m}_1,\mathbf{m}_0}) - \mathbf{g}^f_{\epsilon,c}(\mathbf{m}_0),
\end{align*}
where $\Delta \mathbf{g}^f_{\epsilon,c}$ is the critical height on the modified landscape. We also write $\Delta \mathbf{g}_{\epsilon}$ to denote the critical height on the original landscape. 
Suppose that $\Delta \mathbf{g}_{\epsilon}$ is attained at $\mathbf{m}_4$ so that $\Delta \mathbf{g}_{\epsilon} = \mathbf{g}_{\epsilon}(\mathbf{m}_4) - \mathbf{g}_{\epsilon}(\mathbf{m}_0)$, and we deduce
\begin{align*}
	\Delta \mathbf{g}^f_{\epsilon,c} \leq \mathbf{g}^f_{\epsilon,c}(\mathbf{m}_4) - \mathbf{g}^f_{\epsilon,c}(\mathbf{m}_0) &= \int_{\mathbf{E}(\mathbf{m}_0)}^{\mathbf{E}(\mathbf{m}_4)} \dfrac{1}{f((u-c)_+) + \epsilon} \, du \\
	&\quad + \frac{1}{2}\left(I(2\mathbf{m}_4^+) + I(2\mathbf{m}_4^-)\right) - \frac{1}{2}\left(I(2\mathbf{m}_0^+) + I(2\mathbf{m}_0^-)\right) \\
	&\leq \frac{1}{\epsilon} \left(\mathbf{g}_{\epsilon}(\mathbf{m}_4) - \mathbf{g}_{\epsilon}(\mathbf{m}_0)\right) = \frac{1}{\epsilon} \Delta \mathbf{g}_{\epsilon}.
\end{align*}
In other words, the critical height of the free energy in the modified landscape is bounded above by $\frac{1}{\epsilon}$ times the critical height of the free energy in the original landscape. 

A direct application of \cite[Theorem $2.7$]{MP98} yields the following result on the asymptotics of the spectral gap:
\begin{theorem}[Asymptotics of the spectral gap]
	Suppose $\theta > 0$, $\frac{1}{\epsilon} > \cosh^2\left(\frac{1}{\epsilon} \theta\right)$ are fixed, and $\mathbf{m}_0$ is the saddle point while $\mathbf{m}_1, \mathbf{m}_2$ are the two global minima on the original free energy landscape $\mathbf{g}_{\epsilon}$. For $c \in [\mathbf{E}(\mathbf{m}_1), \mathbf{E}(\mathbf{m}_0)]$, we have, $\mathbb{P}$-almost surely that
	\begin{align*}
		\lim_{N \to \infty} \dfrac{1}{N} \log \lambda_2(-\mathbf{M}^f_{\epsilon,c,N}) &= - \Delta \mathbf{g}^f_{\epsilon,c} \geq - \frac{1}{\epsilon} \Delta \mathbf{g}_{\epsilon} = \lim_{N \to \infty} \dfrac{1}{N} \log \lambda_2(-\mathbf{M}^0_{\epsilon,N}).
	\end{align*}
	In essence, the relaxation time in the mean-field limit of the transformed generator $\mathbf{M}^f_{\epsilon,c,N}$ is asymptotically less than or equal to that of the original generator $\mathbf{M}^0_{\epsilon,N}$.
\end{theorem}
\textcolor{black}{This subsequently gives a reduced exponential dependence of the relaxation time on the modified landscape compared with the relaxation time on the original landscape.}

\section{Discrete simulated annealing with landscape modification}\label{sec:discreteSA}

Unlike previous sections of this paper where the temperature parameter is fixed, in this section we consider the non-homogeneous Metropolis-Hastings with landscape modification where the temperature schedule $(\epsilon_t)_{t \geq 0}$ is time-dependent, non-increasing and goes to zero as $t \to \infty$.

We first recall the concept of critical height as introduced in Section \ref{sec:MHland}. Precisely, we define
\begin{align}
	H^f_{\epsilon,c} &:= L(\mathcal{H}^f_{\epsilon,c}) = \max_{x,y \in \mathcal{X}}\{G(\mathcal{H}^f_{\epsilon,c},x,y) - \mathcal{H}^f_{\epsilon,c}(x) - \mathcal{H}^f_{\epsilon,c}(y) \}, \label{eq:chland}\\
	H^0 &:= L(\mathcal{H}) = \max_{x,y \in \mathcal{X}}\{G(\mcH,x,y) - \mathcal{H}(x) - \mathcal{H}(y)\} + \min \mathcal{H}, \\
	c^* &:= c^*(\mathcal{H},c) = \max_{x,y \in \mathcal{X}} \{(G(\mathcal{H},x,y) \wedge c) - (\mathcal{H}(x) \wedge c) - (\mathcal{H}(y) \wedge c) \} + \min \mathcal{H}, \label{eq:clippedch}
\end{align}
where $H^f_{\epsilon,c}$ is the critical height associated with the modified landscape, $H^0$ is the critical height associated with the original landscape $\mcH$, and $c^*$ is the clipped critical height. We shall see in the main results of this Section below that both $H^f_{\epsilon,c}$ and $c^*$ play a fundamental role in the relaxation time in the low temperature regime and in determining the cooling schedule of an improved simulated annealing algorithm running on the modified landscape. 

\textcolor{black}{As an illustration to calculate and compare these critical heights, we consider a simple one-dimensional landscape with a saddle point at $s$, local (but not global) minimum at $m$ and a single global minimum. At temperature $\epsilon = 1$, the original critical height is attained at $H^0 = \mcH(s) - \mcH(m)$ while the modified critical height is $H^f_{1,c} = \mcH^f_{1,c}(s) - \mcH^f_{1,c}(m) \leq \mcH^0$. In this setting, depending on whether $c$ is above or below $\mcH(m)$, the clipped critcal height $c^*$ is
\begin{align*}
	c^* = \begin{cases}
		0, \quad \textrm{if} \, c \leq \mcH(m)\,, \\
		c - \mcH(m), \quad \textrm{if} \, c > \mcH(m)\,.
	\end{cases}
\end{align*} 
These critical heights are illustrated in Figure \ref{fig:c*}.
}

\begin{figure}
	\centering
	{\renewcommand{\arraystretch}{0}
		\begin{tabular}{c@{}c}
			\begin{subfigure}[b]{.475\columnwidth}
				\centering
				\includegraphics[width=\columnwidth]{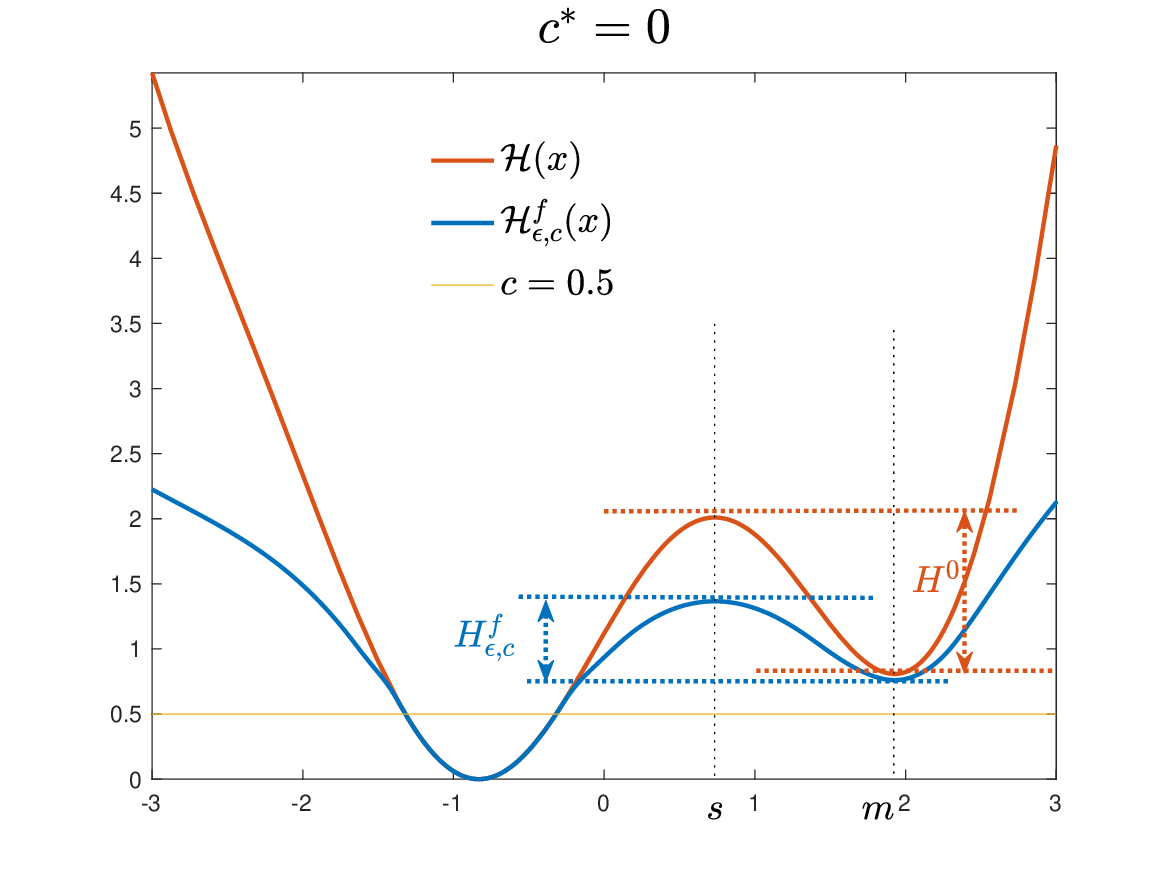}%
			\end{subfigure}&
			\begin{subfigure}[b]{.475\columnwidth}  
				\centering
				\includegraphics[width=\columnwidth]{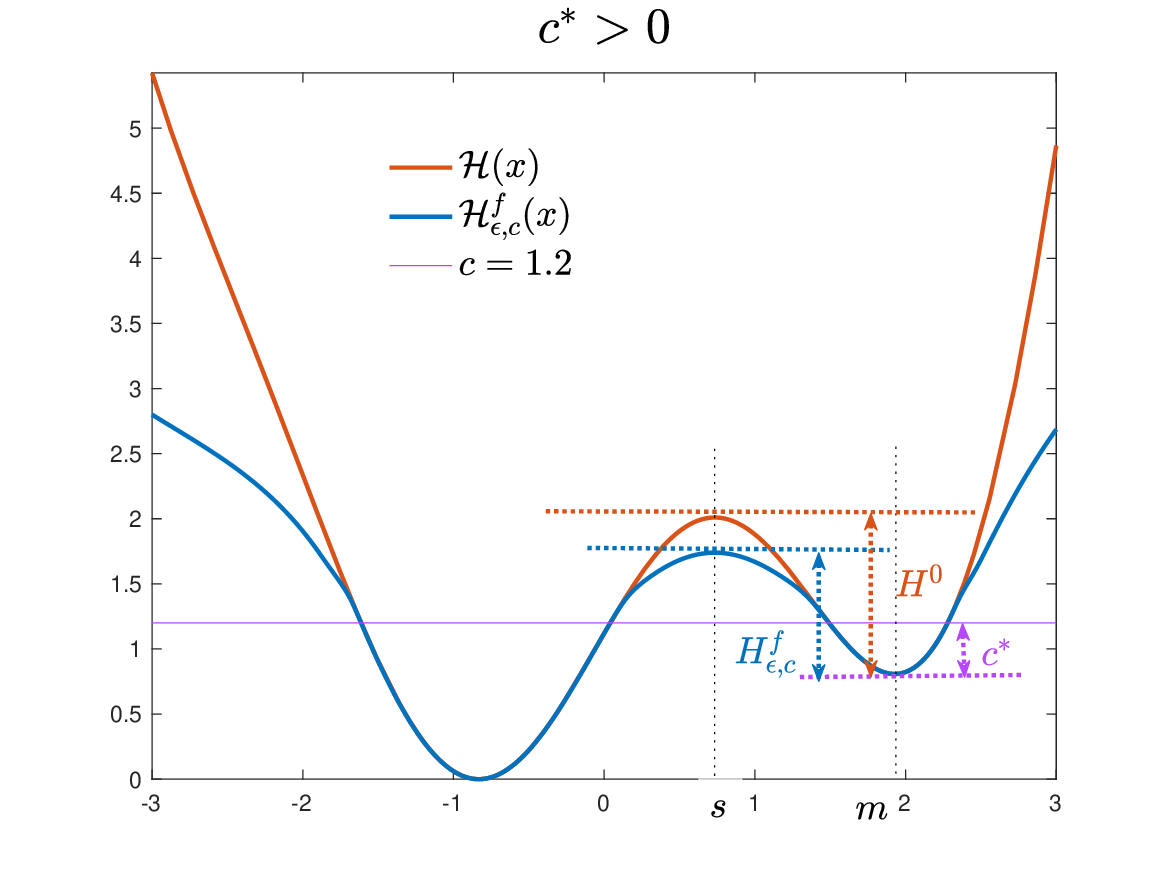}%
			\end{subfigure}
	\end{tabular}}
	\caption{\small Comparing the critical heights $H^0, H^f_{1,c}, c^*$ on a one-dimensional landscape generated by $\mcH$ with a local minimum at $m$ and a global minimum, where we take $\epsilon = 1$, $f(z) = z$ and the proposal chain is of nearest-neighbour type, i.e. going to left or right with probability $1/2$.} 
	\label{fig:c*}
\end{figure}

Our first result gives the asymptotic order of the spectral gap $\lambda_2(-M^f_{\epsilon,c})$ in terms of $c^*$ in the low temperature regime, which will be proven to be essential in obtaining convergence result for simulated annealing:
\begin{theorem}\label{thm:asympspectral}
	Assume that $f$ and $\min \mcH \leq c \leq \max \mcH$ satisfy Assumption \ref{assump:main}, and in addition for all small enough $z > 0$ we have $f(z) \geq z$. There exists positive constants $C_2, C_3, C_4$ that depend on the state space $\mathcal{X}$ and the proposal generator $Q$ but not on the temperature $\epsilon$, and subexponential function 
	$$C_1(\epsilon) := \begin{cases} 
		\frac{1}{C_2} \left(1 + \frac{1}{\epsilon}(\max \mcH - c)\right)\exp\bigg\{\frac{1}{f(\delta)} \left(\max \mcH - \min \mcH\right)\bigg\}, & \mbox{if } c < \max \mcH; \\
		\frac{1}{C_2}, & \mbox{if } c = \max \mcH, 
	\end{cases}$$
	where $\delta := \min_{x;~ \mcH(x) > c} \{\mcH(x) - c\}$, such that
	\begin{align*}
		C_1^{-1}(\epsilon) e^{-\frac{1}{\epsilon} c^*} \leq C_2 e^{-H^f_{\epsilon,c}} \leq \lambda_2(-M^f_{\epsilon,c}) \leq C_3 e^{-H^f_{\epsilon,c}} \leq C_4 e^{-\frac{1}{\epsilon} c^*}, 
	\end{align*}
	where $H^f_{\epsilon,c}$ is introduced in \eqref{eq:chland} and $c^*$ is defined in \eqref{eq:clippedch}. Consequently, this leads to
	\begin{align*}
		\lim_{\epsilon \to 0} \epsilon \log \lambda_2(-M^f_{\epsilon,c}) = - c^*.
	\end{align*}
\end{theorem}

As a corollary of the above result Theorem \ref{thm:asympspectral}, using the asymptotics of the spectral gap we derive similar asymptotics of the mixing time and tunneling time on the modified landscape:
\begin{corollary}[Asymptotics of mixing and tunneling times in the low-temperature regime]\label{cor:mixtun}
	Assume the same setting as in Theorem \ref{thm:asympspectral}. Let $S_{\textrm{min}} := \arg \min \mcH(x)$ be the set of global minima of $\mcH$, $\eta \in S_{\textrm{min}}$ and $\sigma,\eta$ attain $H^0$ such that $H^0 = G^0(\sigma,\eta) - \mcH(\sigma)$. Then the following statements hold:
	\begin{enumerate}
		\item\label{it:cormix}
		$$\lim_{\epsilon \to 0} \epsilon \log t_{mix}(M^f_{\epsilon,c},1/4) = c^*.$$
		\item\label{it:cortun}
		$$\lim_{\epsilon \to 0} \epsilon \log \E_{\sigma}(\tau^f_\eta) = c^* \leq H^0 = \lim_{\epsilon \to 0} \epsilon \log \E_{\sigma}(\tau^0_\eta).$$
	\end{enumerate}
	In particular, when $c = \mcH_{\textrm{min}}$, we have subexponential tunneling time as $\lim_{\epsilon \to 0} \epsilon \log \E_{\sigma}(\tau^f_\eta) = 0 = c^*$.
\end{corollary}
Note that in the case where both $\sigma, \eta \in S_{\textrm{min}}$ with initial state $\sigma$, it is a reasonable choice to pick the parameter $c = \mcH(X^f(0)) = \mcH(\sigma) = \mcH_{\textrm{min}}$, and in this setting we have subexponential tunneling time on the modified landscape. For instance, in applications we may know about a global minimizer $\sigma$, and by setting $c = \mcH(\sigma)$ we can search for other possible global minimizer(s) owing to the subexponential tunneling in the low-temperature regime. For a concrete example, in the Widom-Rowlinson model with $m \in \mathbb{N}$ particle types, $S_{\textrm{min}}$ is precisely the set of configurations in which all sites are occupied by particles of the same type and hence both $S_{\textrm{min}}$ and $\mcH_{\textrm{min}}$ are known in this model. We refer interested readers to \cite{NZ19,Z18} for work on the energy landscape analysis of various statistical physics models in this direction.

To prove convergence result for simulated annealing with landscape modification, as our target function $\mcH^f_{\epsilon,c}$ depends on time through the cooling schedule, we are in the setting of simulated annealing with time-dependent energy function as in \cite{Lowe96}. We first present the following auxillary lemma, where we verify various assumptions in \cite{Lowe96} in our setting. We also decide to put it in this section rather than in the proof since it will help to better understand the convergence result in Theorem \ref{thm:mainsaland} below.

\begin{lemma}\label{lem:sa}
	Assume the same setting as in Theorem \ref{thm:asympspectral}. Let $M := \max \mcH - \min \mcH$, $\beta_t := 1/\epsilon_t$ and the cooling schedule is, for small enough $\epsilon$ such that $M + \max \mcH - c > \epsilon > 0$ and $t \geq 0$,
	$$\epsilon_t = \dfrac{c^* + \epsilon}{\ln(t+1)}.$$
	We have
	\begin{enumerate}
		\item\label{it:timed1} For all $x \in \mathcal{X}$ and all $t \geq 0$, 
		$$0 \leq \epsilon_t \mcH^f_{\epsilon_t,c}(x) \leq M.$$
		
		\item\label{it:timed2} For all $x \in \mathcal{X}$, $$\left|\dfrac{\partial}{\partial t} \epsilon_t \mcH^f_{\epsilon_t,c}(x)\right| \leq \dfrac{2M}{(\ln(1+t))(1+t)}.$$
		
		\item\label{it:timed3} Let $R_t := \sup_x \frac{\partial}{\partial t} \epsilon_t \mcH^f_{\epsilon_t,c}(x)$ and $B := 6M/(c^*+\epsilon)$. For all $t \geq 0$,
		$$\beta_t^{\prime} M + \beta_t R_t \leq \dfrac{3M}{(c^*+\epsilon)(1+t)} = \dfrac{B}{2(1+t)}.$$
		
		\item\label{it:timed4} Let
		$$p := \dfrac{2M}{M + \max \mcH - \epsilon - c} > 2,$$
		and
		$$A := \begin{cases} 
			\dfrac{1}{C_2 (\min_x \mu(x))^{(p-2)/p}}\exp\bigg\{\frac{1}{f(\delta)} \left(\max \mcH - \min \mcH\right)\bigg\}, & \mbox{if } c < \max \mcH; \\
			\dfrac{1}{C_2 (\min_x \mu(x))^{(p-2)/p}}, & \mbox{if } c = \max \mcH, 
		\end{cases}$$
		where $C_2, \delta$ are as in Theorem \ref{thm:asympspectral} and we recall that $\mu$ is the stationary measure of the proposal generator $Q$. For $g \in \ell^p(\pi_{\epsilon_t,c}^f)$, we have
		$$\norm{g - \pi_{\epsilon_t,c}^f(g)}_{\ell^p(\pi^f_{\epsilon_t,c})}^2 \leq A (1+t) \langle -M^f_{\epsilon_t,c}g,g \rangle_{\pi^f_{\epsilon_t,c}}.$$
	\end{enumerate}
\end{lemma}

\begin{rk}
	Item \eqref{it:timed1} and item \eqref{it:timed2} correspond to respectively equation $(2.11)$ and $(2.12)$ in \cite{Lowe96}, while the lower bound of the spectral gap in Theorem \ref{thm:asympspectral} verify equation $(2.13)$ in \cite{Lowe96}.  Assumptions $(A1)$ and $(A2)$ in \cite{Lowe96} are checked in item \eqref{it:timed3} and item \eqref{it:timed4} respectively.
\end{rk}

With the above Lemma and the notations introduced there, we are ready to give one of the main results of this paper concerning the large-time convergence of discrete simulated annealing with landscape modification. The gain of landscape modification in simulated annealing can be seen by operating a possibly faster logarithmic cooling schedule with clipped critical height $c^*$, while in classical simulated annealing the critical height is $H^0$. The possible benefit thus depends on the tuning of $c$ since $c^* \leq c - \mcH_{\textrm{min}}$. This result is analogous to the result that we have obtained in \cite{C20KSA} for improved kinetic simulated annealing. Similar improvement of logarithmic cooling schedule by means of reduction in critical height can be found in the infinite swapping algorithm \cite{MSTW22}.

\begin{theorem}\label{thm:mainsaland}
	Assume the same setting as in Theorem \ref{thm:asympspectral}. Let $A,B,p$ be the quantities as introduced in Lemma \ref{lem:sa}. Define
	\begin{align*}
		\overline{\epsilon} &:= \dfrac{p-2}{p}, \quad K := \dfrac{4(1+2AB)}{1 - \exp\{-\frac{1}{2A} - B\}}, \quad S_{\textrm{min}} := \arg \min \mcH, \quad \underline{d} := \min_{x;~\mcH(x) \neq \mcH_{\textrm{min}}} \mcH(x).
	\end{align*}
	Under the cooling schedule of the form, for any $\epsilon > 0$ as in Lemma \ref{lem:sa},
	$$\epsilon_t = \dfrac{c^* + \epsilon}{\ln(t+1)},$$
	we then have, for any $x \in \mathcal{X} \backslash S_{\textrm{min}}$ and $t \geq e^{1/\overline{\epsilon}}-1$,
	$$\mathbb{P}_x\left(\tau^f_{S_{\textrm{min}}} > t\right) \leq (1+K^{\frac{1}{2\overline{\epsilon}}}) \sqrt{\pi^f_{\epsilon_t,c}(\mathcal{X} \backslash S_{\textrm{min}})}  + \pi^f_{\epsilon_t,c}(\mathcal{X} \backslash S_{\textrm{min}}) \to 0 \quad \text{as } t \to \infty.$$
	Note that
	\begin{align*}
		\pi^f_{\epsilon_t,c}(\mathcal{X} \backslash S_{\textrm{min}}) &\leq \begin{cases} 
			\dfrac{1}{\mu(S_{\textrm{min}})} e^{-\frac{1}{\epsilon_t}(\underline{d} - \mcH_{\textrm{min}})} , & \mbox{if } c \geq \underline{d}; \\
			\dfrac{1}{\mu(S_{\textrm{min}})}\exp\bigg\{-\int_{c}^{\underline{d}} \frac{1}{f(u-c) + \epsilon_t}du\bigg\}, & \mbox{if } \mcH_{\textrm{min}} \leq c < \underline{d}.
		\end{cases}
	\end{align*}
\end{theorem}

\begin{rk}[On tuning the threshold parameter $c$]
	There are various ways to tune the parameter $c$ for improved convergence. In \cite{C20KSA}, we propose to use the running minimum generated by the algorithm to tune $c$. Note that for the Curie-Weiss model, in the second paragraph below Theorem \ref{thm:landmodCW} we have already explained how one can tune the parameter $c$ in that setting.
\end{rk}

\subsection{Numerical illustrations}\label{subsec:numerical}

Before we proceed to discuss the proofs of the main results above, we illustrate and compare the convergence performance of simulated annealing with landscape modification, that we call improved simulated annealing (ISA), against the classical simulated annealing algorithm (SA) on the travelling salesman problem (TSP). We first state the parameters that we used to generate the numerical results:

\textbf{TSP and its objective function}. 50 nodes are uniformly random on the grid $[0,100] \times [0,100]$. The objective is to find a configuration that minimize the total Euclidean distance with the same starting and ending point. Each node can only be visited once.

\textbf{Initial configuration}. Both ISA and SA have the same initialization. They are initialized using the output of the nearest-neighbour algorithm: a node is randomly chosen as the starting point, which is then connected to the closest unvisited node. It repeats until every node has been visited, and subsequently the last node is connected back to the starting node.

\textbf{Proposal chain}. Both ISA and SA share the same proposal chain: at each step, a proposal move is generated using the 2-OPT algorithm \cite{C58}.

\textbf{Acceptance-rejection mechanism}. In SA, the proposed move is accepted with probability \\
$\min \left\{ 1,e^{\beta(\mathcal{H}(x)-\mathcal{H}(y))} \right\}$, while in ISA, the acceptance probability is computed as in Section \ref{subsec:linearf}. In other words, we use a linear $f$ in ISA. Both SA and ISA share the same source of randomness.

\textbf{Cooling schedule}. Both ISA and SA use the same logarithmic cooling schedule of the form
$$\epsilon_t = \dfrac{\sqrt{50}}{\ln(t+1)}.$$

\textbf{Choice of $c$ in ISA}. In this experiment, if we denote the proposal configuration at time $t$ to be $y_t$, we set $c = c_t$ to be
$$c_t = \mathcal{H}(y_t) - 5.$$
This tuning strategy has been discussed in Section \ref{subsec:tuning}.

\textbf{Number of iterations}. We run both ISA and SA for 100,000 iterations.

We generate 1000 random TSP instances. For each instance we compute what we call the improvement percentage ($\textrm{IP}$) of ISA over SA, defined by
$$\textrm{IP} := 100 \dfrac{\min_{t \in [0,100,000]} \mathcal{H}(X^0(t)) - \min_{t \in [0,100,000]} \mathcal{H}(X^f(t))}{\min_{t \in [0,100,000]} \mathcal{H}(X^0(t))}.$$

The summary statistics of $\textrm{IP}$ are provided in Table \ref{table:1}, while its histogram over these 1000 instances can be found in Figure \ref{fig:Histogram_TSP}. The code for reproducing these results can be found in \url{https://github.com/mchchoi/Improved-discrete-simulated-annealing}.

The summary statistics in Table \ref{table:1} and the histogram in Figure \ref{fig:Histogram_TSP} offer empirical evidence in using ISA over SA: out of the 1000 TSP instances, there are 798 instances in which the improvement percentage $\textrm{IP}$ is non-negative. The sample mean of $\textrm{IP}$ is approximately 1.87\% while its sample median is 1.47\%.

Next, we look into a particular instance and investigate the difference between SA and ISA in Figure \ref{fig:Objfunction_TSP}. We see that SA (blue curve) is stuck at a local minimum, while ISA (orange curve) is able to escape the local minimum, owing to the increased acceptance probability compared with SA, and it reaches regions where the objective value is smaller than that of SA.

\begin{table}[H]
	\begin{tabular}{c|c}
			Sample mean	& 1.87\% \\ \hline
			Sample median	& 1.47\% \\ \hline
			Sample maximum	& 11.35\% \\ \hline
			Sample minimum	& -9.21\% \\ \hline
			Numbers of $\textrm{IP} \geq 0$	& 798 \\ \hline
			Numbers of $\textrm{IP} < 0$	& 202
		\end{tabular}
	\caption{Summary statistics of $\textrm{IP}$ on 1000 random TSP instances}
	\label{table:1}
\end{table}

\begin{figure}[H]
	\includegraphics[width=0.6\linewidth]{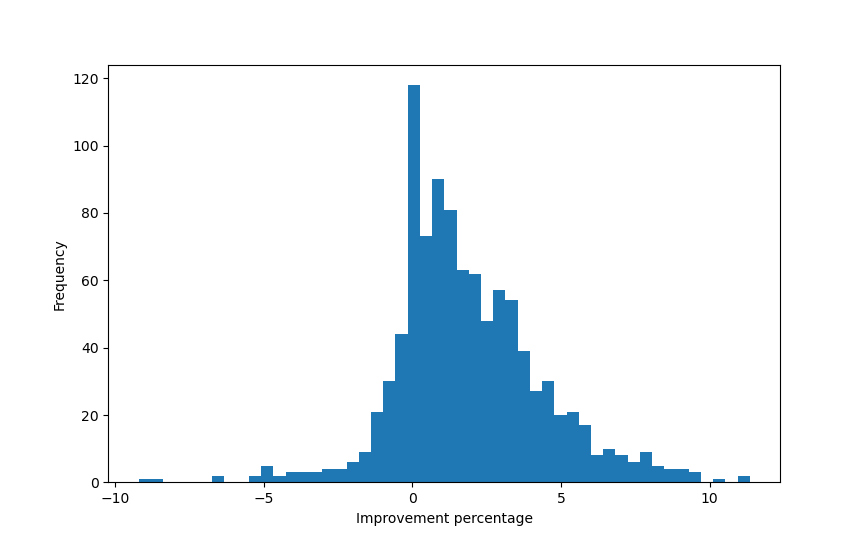}
	\caption{Histogram of improvement percentage of ISA over SA on 1000 randomly generated TSP instances}
	\label{fig:Histogram_TSP}
\end{figure}

\begin{figure}[H]
	\includegraphics[width=0.55\linewidth]{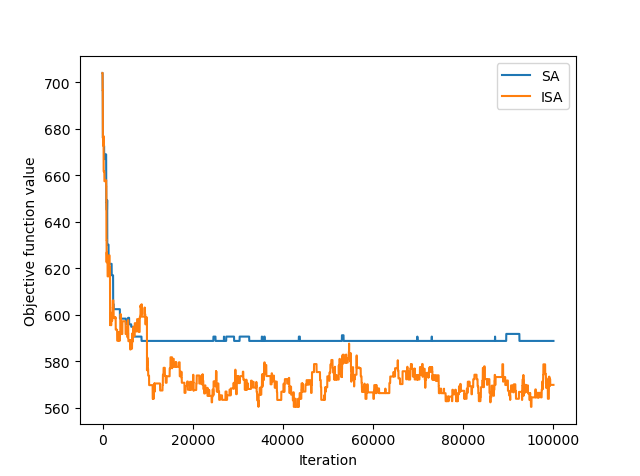}
	\caption{TSP objective value against iteration of ISA and SA}
	\label{fig:Objfunction_TSP}
\end{figure}

The rest of this section is devoted to the proofs of Theorem \ref{thm:asympspectral}, Corollary \ref{cor:mixtun}, Lemma \ref{lem:sa} and Theorem \ref{thm:mainsaland}.

\subsection{Proof of Theorem \ref{thm:asympspectral}}

First, using the classical result by \cite[Theorem $2.1$]{HS88}, it is immediate that
\begin{align*}
	C_2 e^{-H^f_{\epsilon,c}} \leq \lambda_2(-M^f_{\epsilon,c}) \leq C_3 e^{-H^f_{\epsilon,c}}.
\end{align*}
For any arbitrary $x_1, x_2 \in \{\mcH(x_1) \geq \mcH(x_2)\}$, we deduce the following upper bound:
\begin{align*}
	\mathcal{H}^f_{\epsilon,c}(x_1) - \mathcal{H}^f_{\epsilon,c}(x_2) &= \int_{\mcH(x_2)}^{\mathcal{H}(x_1)} \dfrac{1}{f((u-c)_+) + \epsilon}\,du \\ 
	&= \begin{cases} 
		\frac{1}{\epsilon}(\mcH(x_1) - \mcH(x_2)), & \mbox{if } c\geq \mcH(x_1) > \mcH(x_2); \\
		\frac{1}{\epsilon}(c - \mcH(x_2)) + \int_{c}^{\mcH(x_1)} \frac{1}{f(u-c) + \epsilon}\,du, & \mbox{if } \mcH(x_1) > c \geq \mcH(x_2); \\
		\int_{\mcH(x_2)}^{\mcH(x_1)} \frac{1}{f(u-c) + \epsilon}\,du, & \mbox{if } \mcH(x_1) > \mcH(x_2) > c. 
	\end{cases} \\
	&\leq \begin{cases} 
		\frac{1}{\epsilon}(\mcH(x_1) \wedge c - \mcH(x_2) \wedge c), & \mbox{if } c\geq \mcH(x_1) > \mcH(x_2); \\
		\frac{1}{\epsilon}(\mcH(x_1) \wedge c - \mcH(x_2)\wedge c) + \ln(1+\frac{1}{\epsilon}(\max \mcH-c)), & \mbox{if } \mcH(x_1) > c \geq \mcH(x_2); \\
		\frac{1}{f(\delta)} \left(\max \mcH - \min \mcH\right), & \mbox{if } \mcH(x_1) > \mcH(x_2) > c,.
	\end{cases}
\end{align*}
As a result, $C_2^{-1}e^{H^f_{\epsilon,c}} \leq e^{\frac{1}{\epsilon} c^*} C_1(\epsilon)$. On the other hand, we have the following lower bound:
\begin{align*}
	\mathcal{H}^f_{\epsilon,c}(x_1) - \mathcal{H}^f_{\epsilon,c}(x_2) 
	&\geq \begin{cases} 
		\frac{1}{\epsilon}(\mcH(x_1) \wedge c - \mcH(x_2) \wedge c), & \mbox{if } c\geq \mcH(x_1) > \mcH(x_2); \\
		\frac{1}{\epsilon}(\mcH(x_1) \wedge c - \mcH(x_2)\wedge c), & \mbox{if } \mcH(x_1) > c \geq \mcH(x_2); \\
		\frac{1}{f(\mcH(x_1) - c)} \left(\mcH(x_1) - \mcH(x_2)\right), & \mbox{if } \mcH(x_1) > \mcH(x_2) > c, 
	\end{cases}
\end{align*}
and hence $e^{H^f_{\epsilon,c}} \geq e^{\frac{1}{\epsilon} c^*} C_4^{-1}$.

\subsection{Proof of Corollary \ref{cor:mixtun}}

We first prove item \eqref{it:cormix}. For continuous-time reversible Markov chain, by \cite[Theorem $12.5$, $20.6$]{LPW17} we bound the total variation mixing time by relaxation time via
$$\dfrac{1}{\lambda_2(-M^f_{\epsilon,c})} \log 2 \leq t_{mix}(M^f_{\epsilon,c},1/4) \leq \dfrac{1}{\lambda_2(-M^f_{\epsilon,c})} \log\left(\dfrac{4}{\pi^f_{min}}\right),$$
where $\pi_{min}^f := \min_x \pi^f_{\epsilon,c}(x) = \dfrac{\mu(x^*)}{Z^f_{\epsilon,c}},$ for some $x^* \in S_{\textrm{min}}$ and $Z^f_{\epsilon,c} := \sum_{x \in \mathcal{X}} e^{-\mcH_{\epsilon,c}^f(x)}\mu(x)$ is the normalization constant. Note that since $ \ln Z^f_{\epsilon,c} \to \ln \mu(S_{\textrm{min}})$ and so
$$\lim_{\epsilon \to 0} \epsilon \ln Z^f_{\epsilon,c} = 0.$$
Item \eqref{it:cormix} follows by collecting the above results together with Theorem \ref{thm:asympspectral}.

Next, we prove item \eqref{it:cortun}. First, using the random target lemma \cite[Section $4.2$]{AF14}, we have
$$\pi^f_{\epsilon,c}(\eta) \E_{\sigma}(\tau^f_\eta) \leq \sum_{x \in \mathcal{X}} \pi^f_{\epsilon,c}(x) \E_{\sigma}(\tau^f_x) \leq (|\mathcal{X}|-1) \dfrac{1}{\lambda_2(-M^f_{\epsilon,c})}.$$
Since $\mcH(\eta) = 0$ and $Z^f_{\epsilon,c} \leq 1$, upon rearranging and using Theorem \ref{thm:asympspectral} yields
$$\limsup_{\epsilon \to 0} \epsilon \log \E_{\sigma}(\tau^f_\eta) \leq \lim_{\epsilon \to 0} \epsilon \log \dfrac{1}{\lambda_2(-M^f_{\epsilon,c})} = c^*.$$

Define the equilibrium potential and capacity of the pair $(\sigma,\eta)$ as in \cite[Chapter $7.2$]{BH15} to be respectively
\begin{align*}
	h_{\sigma,\eta}^{f}(x) &:= \mP_x(\tau^f_{\sigma} < \tau^f_{\eta}), \\
	\mathrm{cap}^{M^f_{\epsilon,c}}(\sigma,\eta) &:= \inf_{f: f|_{A} = 1, f|_{B} = 0} \langle -M^f_{\epsilon,c}f,f \rangle_{\pi^f_{\epsilon,c}} = \langle -M^f_{\epsilon,c} h_{\sigma,\eta}^{f}, h_{\sigma,\eta}^{f} \rangle_{\pi^f_{\epsilon,c}}.
\end{align*}
If we prove that 
\begin{align}\label{eq:capM2}
	\mathrm{cap}^{M^f_{\epsilon,c}}(\sigma,\eta) \leq \dfrac{1}{Z^f_{\epsilon,c}} \left(\sum_{x,y} \mu(x) Q(x,y)\right) e^{-\frac{1}{\epsilon} G^0(\sigma,\eta)},
\end{align}
then together the mean hitting time formula with equilibrium potential and capacity leads to
\begin{align*}
	\E_{\sigma}(\tau^f_{\eta}) &= \dfrac{1}{\mathrm{cap}^{M^f_{\epsilon,c}}(\sigma,\eta)} \sum_{y \in \mathcal{X}} \pi_{\epsilon,c}^f(y) h_{\sigma,\eta}^{f}(y) \\
	&\geq \dfrac{1}{\mathrm{cap}^{M^f_{\epsilon,c}}(\sigma,\eta)} \pi_{\epsilon,c}^f(\sigma) h_{\sigma, \eta}^{f}(\sigma) \geq \dfrac{\mu(\sigma)}{\sum_{x,y} \mu(x) Q(x,y)} e^{\frac{1}{\epsilon} c^*},
\end{align*}
and the desired result follows since
$$\liminf_{\epsilon \to 0} \epsilon \log \E_{\sigma}(\tau^f_\eta) \geq c^*.$$
It therefore remains to prove \eqref{eq:capM2}. Define
$$\Phi(\sigma,\eta) := \{x\in\mathcal{X};~ G^0(x, \sigma) \leq G^0(x,\eta)\}.$$
Writing $\mathbf{1}_{A}$ to be the indicator function of the set $A$, the Dirichlet principle of capacity gives
\begin{align*}
	\mathrm{cap}^{M^f_{\epsilon,c}}(\sigma,\eta) \leq \langle -M^f_{\epsilon,c}\mathbf{1}_{\Phi(\sigma,\eta)},\mathbf{1}_{\Phi(\sigma,\eta)}\rangle_{\pi^f_{\epsilon,c}} &= \dfrac{1}{Z^f_{\epsilon,c}}\sum_{x \in \Phi(\sigma,\eta), y \notin \Phi(\sigma,\eta)}  e^{-\frac{1}{\epsilon} \left( H(x) \vee H(y) \right)} \mu(x) Q(x,y) \\
	&\leq \dfrac{1}{Z^f_{\epsilon,c}} \left(\sum_{x,y} \mu(x) Q(x,y)\right) e^{-\frac{1}{\epsilon} G^0(\sigma,\eta)}.
\end{align*}
where in the last inequality we use the fact that $G^0(\sigma,\eta)$ is the lowest possible highest elevation between $\sigma$ and $\eta$.

\subsection{Proof of Lemma \ref{lem:sa}}

We first prove item \eqref{it:timed1}. The lower bound is immediate, while the upper bound can be deduced via
$$\epsilon_t \mcH^f_{\epsilon_t,c} \leq \epsilon_t \int_{\mcH_{\textrm{min}}}^{\mcH(x)} \dfrac{1}{\epsilon_t} \,du \leq M.$$

Next, we prove item \eqref{it:timed2}. We consider
\begin{align*}
	\dfrac{\partial}{\partial t} \epsilon_t \mcH^f_{\epsilon_t,c}(x) 
	&= 	\mcH^f_{\epsilon_t,c}(x) \left(\dfrac{\partial}{\partial t} \epsilon_t\right) + \epsilon_t \dfrac{\partial}{\partial t} \mcH^f_{\epsilon_t,c}(x) \\
	&= \mcH^f_{\epsilon_t,c}(x) \dfrac{-c^* - \epsilon}{(\ln(t+1))^2} \dfrac{1}{t+1} + \epsilon_t \epsilon_t^{\prime} \int_{\mcH_{\textrm{min}}}^{\mcH(x)} - \dfrac{1}{(f((u-c)_+) + \epsilon_t)^2} \, du.
\end{align*}
This leads to
\begin{align*}
	\left|\dfrac{\partial}{\partial t} \epsilon_t \mcH^f_{\epsilon_t,c}(x)\right| \leq \dfrac{2M}{(\ln(t+1))(t+1)}.
\end{align*}

Thirdly, we prove item \eqref{it:timed3}, and using item \eqref{it:timed2} we calculate that
\begin{align*}
	\beta_t^{\prime} M + \beta_t R_t &\leq \dfrac{M}{(c^*+\epsilon)(t+1)} + \dfrac{2M}{(c^*+\epsilon)(1+t)} = \dfrac{3M}{(c^*+\epsilon)(1+t)}.
\end{align*}

Finally, we prove item \eqref{it:timed4}. Following exactly the same calculation as in the proof of \cite[Lemma $3.5$]{Lowe96}, we see that
\begin{align*}
	\norm{g - \pi_{\epsilon_t,c}^f(g)}_{\ell^p(\pi^f_{\epsilon_t,c})}^2 &\leq A   \left(1 + \beta_t(\max \mcH - c)\right) e^{\beta_t(c^* + M \frac{p-2}{p})}\langle -M^f_{\epsilon_t,c}g,g \rangle_{\pi^f_{\epsilon_t,c}} \\
	&\leq A e^{\beta_t(c^* + M \frac{p-2}{p} + \max \mcH - c)}\langle -M^f_{\epsilon_t,c}g,g \rangle_{\pi^f_{\epsilon_t,c}}.
\end{align*}
The desired result follows if we let
$$\epsilon = \dfrac{M(p-2)}{p} + \max \mcH - c$$
so that
$$p = \dfrac{2M}{M + \max \mcH - \epsilon - c}.$$

\subsection{Proof of Theorem \ref{thm:mainsaland}}

We would like to invoke the results in \cite{Lowe96} for time-dependent target function in simulated annealing. In Lemma \ref{lem:sa} item \eqref{it:timed1}, \eqref{it:timed2}, \eqref{it:timed3} and \eqref{it:timed4}, we verify that equation $(11)$, $(12)$ and Assumption $(A1)$, $(A2)$ respectively hold in \cite{Lowe96}. Consequently, if we let $h_t(y) := \mathbb{P}_x(X^f_{\epsilon_t,c}(t) = y)/\pi^f_{\epsilon_t,c}(y)$, then according to \cite[Lemma $1.7$]{HS88}, its $\ell^2$ norm is bounded by
$$\norm{h_t}_{\ell^2(\pi^f_{\epsilon_t,c})} \leq 1+K^{\frac{1}{2\overline{\epsilon}}},$$
for $t \geq e^{1/\overline{\epsilon}}-1$. The desired result follows from exactly the same argument as in \cite[Theorem $3.8$]{Lowe96}.

Now, we calculate $\pi^f_{\epsilon_t,c}(\mathcal{X} \backslash S_{\textrm{min}})$. For $c \geq \underline{d}$, we compute that
\begin{align*}
	\pi^f_{\epsilon_t,c}(\mathcal{X} \backslash S_{\textrm{min}}) &= \dfrac{\sum_{x \in \mathcal{X} \backslash S_{\textrm{min}}} e^{-\mcH_{\epsilon_t,c}^f(x) }\mu(x)}{\sum_{x \in \mathcal{X}} e^{-\mcH_{\epsilon_t,c}^f(x)}\mu(x)} \\
	&\leq \dfrac{\sum_{x;~ \mcH(x) \geq \underline{d}} \mu(x)}{\sum_{x; \mcH(x) \leq \underline{d}} \exp\{\int_{\mcH(x)}^{\underline{d}} \frac{1}{\epsilon_t}\,du\}\mu(x)} \\
	&\leq e^{-\frac{1}{\epsilon_t}(\underline{d} - \mcH_{\textrm{min}})} \dfrac{1}{\mu(S_{\textrm{min}})}.
\end{align*}
On the other hand, for $\mcH_{\textrm{min}} \leq c < \underline{d}$,
\begin{align*}
	\pi^f_{\epsilon_t,c}(\mathcal{X} \backslash S_{\textrm{min}}) &= \dfrac{\sum_{x \in \mathcal{X} \backslash S_{\textrm{min}}} e^{-\mcH_{\epsilon_t,c}^f(x) }\mu(x)}{\sum_{x \in \mathcal{X}} e^{-\mcH_{\epsilon_t,c}^f(x)}\mu(x)} \\
	&\leq \dfrac{1}{\mu(S_{\textrm{min}})}\exp\bigg\{-\int_{\mcH_{\textrm{min}}}^{\underline{d}} \dfrac{1}{f((u-c)_+) + \epsilon_t}du\bigg\} \\
	&\leq \dfrac{1}{\mu(S_{\textrm{min}})}\exp\bigg\{-\int_{c}^{\underline{d}} \dfrac{1}{f(u-c) + \epsilon_t}du\bigg\}.
\end{align*}

\section*{Acknowledgements}

We thank Laurent Miclo for pointers to the work of Olivier Catoni, and the two reviewers for careful reading and constructive feedback. The author acknowledges the financial support from the startup grant of National University of Singapore, Yale-NUS College and a Singapore MoE Tier 1 grant entitled "MAPLE".

\bibliographystyle{abbrvnat}
\bibliography{thesis}

\end{document}